\documentclass[sigconf,screen]{acmart}

\acmYear{2020}
\acmISBN{} 
\acmDOI{} 
\startPage{1}

\acmConference[]{}{}{}
\acmDOI{}
\acmISBN{}
\acmConference[]{}{}{}

\setcopyright{none}

\usepackage{booktabs}   
\usepackage{subcaption} 

\usepackage{ulem} 

\usepackage{listings}

\usepackage{balance}

\usepackage{bchart}

\usepackage{algorithm}
\usepackage{algpseudocode}

\usepackage{tikz}
\usetikzlibrary{tikzmark}
\usetikzlibrary{arrows,automata}
\usetikzlibrary{trees,shapes,decorations}

\usepackage{microtype}

\renewcommand{\emph}[1]{{\it #1}}

\makeatletter
\newcommand{\xRightarrow}[2][]{\ext@arrow 0359\Rightarrowfill@{#1}{#2}}
\makeatother

\newcommand{\bi}{\begin{array}[t]{@{}l@{}}}
\newcommand{\ei}{\end{array}}
\newcommand{\ba}{\begin{array}}
\newcommand{\ea}{\end{array}}
\newcommand{\bda}{\[\ba}
\newcommand{\eda}{\ea\]}
\newcommand{\bp}{\begin{quote}\tt\begin{tabbing}}
\newcommand{\ep}{\end{tabbing}\end{quote}}

\newcommand{\ignore}[1]{}

\newcommand{\ms}[1]{{\bf MS: #1}}

\newcommand{\mathem}{\sf}

\newcommand{\fig}[3]
        {\begin{figure*}[t]#3\
            \caption{\label{#1}#2} \end{figure*}}

\newcommand{\figurebox}[1]
        {\fbox{\begin{minipage}{\textwidth} #1 \end{minipage}}}
\newcommand{\boxfig}[3]
        {\begin{figure*}\figurebox{#3\caption{\label{#1}#2}}\end{figure*}}

\newcommand{\myirule}[2]{{\renewcommand{\arraystretch}{1.2}\ba{c} #1
                      \\ \hline #2 \ea}}

\newcommand{\rlabel}[1]{\mbox{(#1)}}

\newcommand\Angle[1]{\langle#1\rangle}

\newcommand{\thread}[2]{#1 \sharp #2}

\newcommand{\incC}[2]{{\mathem inc}(#1,#2)}

\newcommand{\maxN}[2]{{\mathem max}(#1,#2)}

\newcommand{\pp}{\ \texttt{++}}

\newcommand{\semP}[3]{#1 \xRightarrow{#2} #3}

\newcommand{\lockE}[1]{\mathit{acq}(#1)} 
\newcommand{\unlockE}[1]{\mathit{rel}(#1)} 
\newcommand{\readE}[1]{r(#1)}
\newcommand{\writeE}[1]{w(#1)}

\newcommand{\readEE}[2]{r(#1)_{#2}}
\newcommand{\writeEE}[2]{w(#1)_{#2}}
\newcommand{\lockEE}[2]{acq(#1)_{#2}}  
\newcommand{\unlockEE}[2]{rel(#1)_{#2}} 

\newcommand{\proj}[2]{\mathit{proj}_{\sharp #1}(#2)}
\newcommand{\pos}[1]{\mathit{pos}(#1)}
\newcommand{\posP}[2]{\mathit{pos}_{{\scriptstyle #1}}(#2)}
\newcommand{\compTID}[1]{\mathit{thread}(#1)}
\newcommand{\compTIDP}[2]{\mathit{thread}_{{\scriptstyle #1}}(#2)}

\newcommand{\events}[1]{\mathit{events}(#1)}

\newcommand{\supVC}[2]{#1 \sqcup #2}
\newcommand{\threadVC}[1]{\mathit{Th}(#1)}

\newcommand{\lastWriteVC}[1]{\mathit{L_W}(#1)}
\newcommand{\lastWriteVCt}[1]{\mathit{L_{W_t}}(#1)}
\newcommand{\lastWriteVCL}[1]{\mathit{L_{W_L}}(#1)}
\newcommand{\concEvt}[1]{\mathit{conc}(#1)}
\newcommand{\accVC}[2]{#1[ #2 ]}
\newcommand{\rwVC}[1]{\mathit{RW}(#1)}

\newcommand{\vcEvt}{\mathit{evt}}
\newcommand{\edges}[1]{\mathit{edges}(#1)}

\newcommand{\gtEdge}{\prec}

\newcommand{\cs}[2]{\thread{#1}{\Angle{#2}}}
\newcommand{\csSym}[1]{CS(#1)}
\newcommand{\csSymAny}{CS}

\newcommand{\acq}[1]{acq(#1)}
\newcommand{\rel}[1]{rel(#1)}

\newcommand{\LS}[1]{LS(#1)}               

\newcommand{\LSt}[1]{LS_t(#1)}   
\newcommand{\LStSym}{LS_t}
\newcommand{\Acq}[1]{Acq(#1)}
\newcommand{\LH}[1]{H(#1)} 

\newcommand{\replay}[3]{\semP{#1}{#2}{#3}}
\newcommand{\sepThree}[3]{(#1 \mid #2 \mid #3)}
\newcommand{\acqLast}{\mathem A}
\newcommand{\writeLast}{\mathem W}
\newcommand{\writeLastMap}{\mathem M}

\newcommand{\TS}{{\mathcal T}}

\newcommand{\rwT}[1]{T^{rw}_{#1}}
\newcommand{\rwTx}{\rwT{x}}

\newcommand{\rwTr}{T^{rw}}

\newcommand{\relTy}{T^{rel}_y}

\newcommand{\hb}[2]{#1 <^{\scriptscriptstyle HB} #2}
\newcommand{\shb}[2]{#1 <^{\scriptscriptstyle SHB} #2}
\newcommand{\wcp}[2]{#1 <^{\scriptscriptstyle WCP} #2}
\newcommand{\wdp}[2]{#1 <^{\scriptscriptstyle WDP} #2}

\newcommand{\pwr}[2]{#1 <^{\scriptscriptstyle PWR} #2}

\newcommand{\pwrs}[2]{#1 <^{\scriptscriptstyle PWRS} #2}
\newcommand{\pwra}[2]{#1 <^{\scriptscriptstyle PWRA} #2}
\newcommand{\pwrr}[2]{#1 <^{\scriptscriptstyle PWRR} #2}

\newcommand{\pwrSym}{<^{\scriptscriptstyle PWR}}
\newcommand{\wrdSym}{<^{\scriptscriptstyle WRD}}

\newcommand{\shbSym}{<^{\scriptscriptstyle SHB}}
\newcommand{\hbSym}{<^{\scriptscriptstyle HB}}
\newcommand{\sdpSym}{<^{\scriptscriptstyle SDP}}
\newcommand{\wdpSym}{<^{\scriptscriptstyle WDP}}
\newcommand{\poSym}{<^{\scriptscriptstyle PO}}
\newcommand{\wcpSym}{<^{\scriptscriptstyle WCP}}
\newcommand{\wcpSymSpecial}[1]{<^{\scriptscriptstyle WCP#1}}

\newcommand{\porhbcSym}{<^{\scriptstyle PORHBC}}

\newcommand{\pwrsSym}{<^{\scriptscriptstyle PWRS}}
\newcommand{\pwraSym}{<^{\scriptscriptstyle PWRA}}
\newcommand{\pwrrSym}{<^{\scriptscriptstyle PWRR}}

\newcommand{\allConcsP}[2]{{\mathcal C}^{#1}(#2)}
\newcommand{\accConcs}[1]{{PC}(#1)}

\newcommand{\allPredRacesP}[1]{{\mathcal P}^{#1}}
\newcommand{\PotentialP}[2]{{\mathcal R}^{#1}_{\not #2}}

\newcommand{\PotentialPWR}[1]{{\mathcal R}^{#1}_{\pwrSym}}

\newcommand{\schedSpecificPredRacesP}[1]{{\mathcal S}^{#1}}

\newcommand{\prefixOf}[2]{#1 \rhd #2}
\newcommand{\dataRace}[4]{#3 \stackrel{\prefixOf{#1}{#2}}{\asymp} #4}

\newcommand{\SHBEE}{\mbox{SHB$^{\scriptstyle E+E}$}}
\newcommand{\SHBEELimit}{\mbox{SHB$_{\scriptstyle L}^{\scriptstyle E+E}$}}

\newcommand{\PWREE}{\mbox{PWR$^{\scriptstyle E+E}$}}
\newcommand{\PWREELimit}{\mbox{PWR$_{\scriptstyle L}^{\scriptstyle E+E}$}}

\newcommand{\PWRZero}{\mbox{PWR}$_{\scriptstyle L}$}

\newcommand{\TSANWRD}{\mbox{TSanWRD}}

\newcommand{\PWR}{\mbox{PWR}}

\newcommand{\NoFP}{\mbox{FP$_{\not\exists}$}}
\newcommand{\OnlyFP}{\mbox{FP$_{\forall}$}}
\newcommand{\NoFN}{\mbox{FN$_{\not\exists}$}}
\newcommand{\OnlyFN}{\mbox{FN$_{\forall}$}}

\makeatletter
\newdimen\legendxshift
\newdimen\legendyshift
\newcount\legendlines
\newcommand{\bclldist}{1mm}
\newcommand{\bclegend}[3][10mm]{%
	\legendxshift=0pt\relax
	\legendyshift=0pt\relax
	\xdef\legendnodes{}%
	\foreach \lcolor/\ltext [count=\ll from 1] in {#3}%
	{\global\legendlines\ll\pgftext{\setbox0\hbox{\bcfontstyle\ltext}\ifdim\wd0>\legendxshift\global\legendxshift\wd0\fi}}%
	\@tempdima#1\@tempdima0.5\@tempdima
	\pgftext{\bcfontstyle\global\legendxshift\dimexpr\bcwidth-\legendxshift-\bclldist-\@tempdima-0.72em}
	\legendyshift\dimexpr5mm+#2\relax
	\legendyshift\legendlines\legendyshift
	\global\legendyshift\dimexpr\bcpos-2.5mm+\bclldist+\legendyshift
	\begin{scope}[shift={(\legendxshift,\legendyshift)}]
		\coordinate (lp) at (0,0);
		\foreach \lcolor/\ltext [count=\ll from 1] in {#3}%
		{
			\node[anchor=north, minimum width=#1, minimum height=5mm,fill=\lcolor] (lb\ll) at (lp) {};
			\node[anchor=west] (l\ll) at (lb\ll.east) {\bcfontstyle\ltext};
			\coordinate (lp) at ($(lp)-(0,5mm+#2)$);
			\xdef\legendnodes{\legendnodes (lb\ll)(l\ll)}
		}
		\node[draw, inner sep=\bclldist,fit=\legendnodes] (frame) {};
	\end{scope}
}
\makeatother

\begin{document}

\title[]{
  Efficient, Near Complete and Often Sound Hybrid Dynamic Data Race Prediction}         




\author{Martin Sulzmann}
\affiliation{
  \institution{Karlsruhe University of Applied Sciences}
  \streetaddress{Moltkestrasse 30}
  \city{Karlsruhe}
  \postcode{76133}
  \country{Germany}
}
\email{martin.sulzmann@gmail.com}

\author{Kai Stadtm{\"u}ller}
\affiliation{
  \institution{Karlsruhe University of Applied Sciences}
  \streetaddress{Moltkestrasse 30}
  \city{Karlsruhe}
  \postcode{76133}
  \country{Germany}
}
\email{kai.stadtmueller@live.de}

\begin{abstract}
  Dynamic data race prediction aims to identify races based
  on a single program run represented by a trace.
  The challenge is to remain efficient while being as sound and as complete as possible.
  Efficient means a linear run-time as otherwise the method unlikely
  scales for real-world programs.
  We introduce an efficient, near complete and often sound
  dynamic data race prediction method that combines
  the lockset method with several improvements made
  in the area of happens-before methods.
  By near complete we mean that the method is complete in theory
  but for efficiency reasons the implementation applies some optimizations
  that may result in incompleteness. The method can be shown to be sound
  for two threads but is unsound in general.
  Experiments show that our method works
  well in practice.
\end{abstract}



\maketitle

\section{Introduction}

We consider verification methods in the context
of concurrently executing programs that make use of multiple threads,
shared reads and writes, and acquire/release operations to protect critical sections.
Specifically, we are interested in data races.
A data race arises if two unprotected, conflicting read/write operations from different
threads happen at the same time.

Detection of data races via traditional run-time testing methods where we simply
run the program and observe its behavior can be tricky.
Due to the highly non-deterministic behavior of concurrent programs, a data race may only
arise under a specific schedule. Even if we are able to force the program to follow
a specific schedule, the two conflicting events many not not happen at the same time.
Static verification methods, e.g.~model checking, are able to explore the entire state space
of different execution runs and their schedules. The issue is that static methods often do not scale
for larger programs. To make them scale, the program's behavior typically needs to be approximated
which then results in less precise analysis results.

The most popular verification method to detect data races combines idea from run-time testing and static verification.
Like in case of run-time testing, a specific program run is considered.
The operations that took place are represented as a program trace.
A trace reflects the interleaved execution of the program run and forms the basis for further analysis.
The challenge is to predict if two conflicting operations may happen at the same time
even if these operations may not necessarily appear in the trace right next to each other.
This approach is commonly referred to as \emph{dynamic data race prediction}.

{\bf Run-Time Events and Traces.}
For example, consider the following trace
  \bda{lll}
  & \thread{1}{} & \thread{2}{}
  \\ \hline
  1. & \writeE{x} &
  \\ 2. & \lockE{y} &
  \\ 3. & \unlockE{y} &
  \\ 4. & & \lockE{y}
  \\ 5. & & \writeE{x}
  \\ 6. & & \unlockE{y}
  \eda
  where for each thread we introduce
a separate column and the trace position can be identified via
 the row number.
 Events $\writeE{x}/\readE{x}$ refer to write/read events on
 the shared variable~$x$.
 Events $\lockE{y}/\unlockE{y}$ refer to acquire/release events on
 lock variable~$y$.
 To identify an event, we often annotate the event with its thread id
 and row number. For example, $\thread{1}{\writeEE{x}{1}}$
 refers to the write event in thread~$1$ at trace position~$1$.
 We sometimes omit the thread id as the trace position (row number)
 is sufficient to unambiguously identify an event.

{\bf Conflicting Events and Data Race Prediction.}
Let $e, f$ be two read/write events on the same variable
where at least one of them is a write event
and both events result from different threads.
Then, we say that $e$ and $f$ are two \emph{conflicting events}.
For the above trace, we find that
$\thread{1}{\writeEE{x}{1}}$ and $\thread{2}{\writeEE{x}{5}}$
are two conflicting events.
Based on the trace we wish to predict if
two conflicting events can appear right next to each other.
Such a situation represents a \emph{data race}.

In the above trace, the two conflicting events
$\writeEE{x}{1}$ and $\writeEE{x}{5}$
do not appear right next to each other in the trace.
Hence, it seems that both events are not in a race.
The point is that a trace represents \emph{one} possible interleaving of concurrent events
but there may be other \emph{alternative} interleavings that result from scheduling
the events slightly differently.
The challenge of data race prediction is to find an alternative interleaving
of the trace such two conflicting events appear right next to each other.

We could explore alternative interleavings by considering
all trace reorderings, i.e.~all permutations of events in the trace.
In general, this is
(a) too inefficient, and (b) leads to false results as the data race
may not be reproducible by re-running the program.
As we only consider the trace and not the program we
impose the following assumptions on a \emph{correctly reordered} trace.
(1) The program order as found in each thread
is respected. (2) Every read sees the same (last) write.
(3) The lock semantics is respected so that execution will not get stuck.

For our running example,
$[\thread{2}{\lockEE{y}{4}}, \thread{2}{\writeEE{x}{5}}, \thread{1}{\writeEE{x}{1}}]$
is a correctly reordered prefix.
We use here list notation to represent the trace.
This reordered trace
serves as a witness for the data race among
the two conflicting events $\writeEE{x}{1}$ and $\writeEE{x}{5}$.
We consider prefixes as we can 'stop' the trace as soon as the
two conflicting events have appeared right next to each other.

{\bf First versus Subsequent Races.}
Earlier works~\citep{Smaragdakis:2012:SPR:2103621.2103702,Kini:2017:DRP:3140587.3062374}
only consider the first race based on a total order of
the occurrence of events in the original trace.
One reason is that a subsequent race may only show itself due to
an earlier race. As the program behavior may be undefined
after the first race, the subsequent race many not be reproducible.

However, it is easy to fix the first race by making
the events mutually exclusive.
  The former subsequent race becomes a first.
  To discover this race we would need to re-run the analysis.
  Hence, it is sensible to report all races
  and not only the first race.

Here is an example to illustrate this point.
\bda{ccc}
  \ba{lll}
  & \thread{1}{} & \thread{2}{}
  \\ \hline
  1. & \readE{y} &
  \\ 2. & \readE{x} &
  \\ 3. && \writeE{y}
  \\ 4. && \writeE{x}
  \ea
  & \mbox{} ~~~~~~ \mbox{}
  &
   \ba{lll}
  & \thread{1}{} & \thread{2}{}
  \\ \hline
  1. & \lockE{y'} &
  \\ 2. & \readE{y} &
  \\ 3. & \unlockE{y'} &
  \\ 4. & \readE{x} &
  \\ 5. && \lockE{y'}
  \\ 6. && \writeE{y}
  \\ 7. && \unlockE{y'}
  \\ 4. && \writeE{x}
  \ea
  \eda
  For the trace on the left,
  $\readEE{y}{1}$ and $\writeEE{y}{3}$ are in a race as shown
  by $[\thread{1}{\readEE{y}{1}}, \thread{2}{\writeEE{y}{3}}]$.

  What about $\readEE{x}{2}$ and $\writeEE{x}{4}$?
  For any reordering where $\readEE{x}{2}$ and $\writeEE{x}{4}$ appear right next to each
  other we find that earlier in the trace
  $\readEE{y}{1}$ and $\writeEE{y}{3}$ appear right next to each other.
  For instance, consider
  $[\thread{1}{\readEE{y}{1}}, \thread{2}{\writeEE{y}{3}},
   \thread{2}{\readEE{x}{2}}, \thread{1}{\writeEE{x}{4}}]$.
  Hence, $\readEE{x}{2}$ and $\writeEE{x}{4}$ represent a subsequent race.

  We can easily fix the first race by making the events involved mutually exclusive.
  See the trace on the right.
  The subsequent race becomes now a first race.

{\bf Our Goals and Contributions.}
For a given trace $T$,
we wish to identify {\emph all predictable} data races in $T$.
This includes first and subsequent races as well.
We write $\allPredRacesP{T}$ to denote the set of all
{\emph predictable data race pairs} $(e,f)$ resulting from $T$
where $e$, $f$ are conflicting events in $T$ and there exists
a correctly reordered prefix of $T$ under which $e$, $f$ appear
right next to each other.

The challenge is to be \emph{efficient}, \emph{sound} and \emph{complete}.
By efficient we mean a run-time that is linear in terms of the size of the trace.
Sound means that races reported by the algorithm can be observed via some
appropriate reordering of the trace.
If unsound, we refer to wrongly a classified race as a \emph{false positive}.
Complete means that all valid reorderings that exhibit some race can be
predicted by the algorithm.
If incomplete, we refer to any not reported race  as a \emph{false negative}.

In this paper, we make the following contributions:
\begin{itemize}
\item We propose an efficient dynamic race prediction method
  that combines the lockset method with
  the happens-before method.
  Our method is novel and
  improves the state-of-the art.
  The method is shown to be complete in general and sound for the case of two threads
  (Section~\ref{sec:shb-wcp-lockset}).
\item We give a detailed description of how to implement our proposed method
      (Section~\ref{sec:implementation}).
  We present an algorithm that overall has quadratic run-time.
  This algorithm can be turned into a linear run-time algorithm
  by sacrificing completeness. For practical as well as contrived examples,
  incompleteness is rarely an issue.
   \item We carry out extensive experiments covering a large
     set of real-world programs as well as a collection of the many challenging
     examples that can be found in the literature.
     For experimentation, we have implemented our algorithm as well as its contenders
      in a common framework.
     We measure the performance, time and space behavior,
     as well as the precision, e.g.~ratio of false positives/negatives etc.
     Measurements show that our algorithm performs well
     compared to state-of-the art algorithms such as ThreadSanitizer, FastTrack, SHB and WCP (Section~\ref{sec:experiments}).
\end{itemize}

The upcoming section gives an overview of our work
and includes also a comparison against closely related works.
Section~\ref{sec:related-work} summarizes related work.
Section~\ref{sec:conclusion} concludes.
The appendix contains additional material such as proofs, extended examples,
optimization details etc.

\section{Happens-Before and Lockset}

We review earlier efficient data race prediction methods
and discuss their limitations.

{\bf Happens-Before Methods.}
The idea is to
is to derive from the trace a happens-before relation among events.
If for two conflicting events, neither event happens before the other event,
this is an indication that both events can appear next to each other.
Happens-before methods can be implemented efficiently
via the help of vector clocks~\citep{Fidge:1991:PAT:646210.683620,Mattern89virtualtime}.
However,
none of the existing happens-before relations~\citep{lamport1978time,Mathur:2018:HFR:3288538:3276515,Kini:2017:DRP:3140587.3062374}
is sound and complete.

For example, Lamport's happens-before relation~\citep{lamport1978time},
referred to as the HB relation,
is neither sound nor complete as shown by the following example.

\begin{example}
  \label{ex:hb-incomplete}
  \label{ex:hb-unsound}
Consider the following two traces.
\bda{lcl}
\ba{l}
\mbox{Trace A}
 \\
  \ba{lll}
  & \thread{1}{}
  & \thread{2}{}
  \\ \hline
  1. & \writeE{x} &
  \\ 2. & \lockE{y} &
  \\ 3. & \unlockE{y} \tikzmark{hb-incomplete-1} &
  \\ 4. & & \tikzmark{hb-incomplete-2} \lockE{y}
  \\ 5. & & \writeE{x}
  \\ 6. & & \unlockE{y}
  \ea
           \begin{tikzpicture}[overlay, remember picture, yshift=.25\baselineskip, shorten >=.5pt, shorten <=.5pt]
           \draw[->] ({pic cs:hb-incomplete-1}) [bend right] node [below, yshift=-0.3cm, xshift=-0.1cm]{\footnotesize{HB}} to ({pic cs:hb-incomplete-2});
           \end{tikzpicture}
  \ea
  &
 &
 \ba{l}
  \mbox{Trace B}
  \\
    \ba{lll}
  & \thread{1}{}
  & \thread{2}{}
    \\ \hline
  1. && \writeE{y}
  \\ 2. & \writeE{x} &
  \\ 3. & \writeE{y} \tikzmark{shb-sound-1} &
  \\ 4. & & \tikzmark{shb-sound-2} \readE{y}
  \\ 5. & & \writeE{x}
  \ea
           \begin{tikzpicture}[overlay, remember picture, yshift=.25\baselineskip, shorten >=.5pt, shorten <=.5pt]
           \draw[->] ({pic cs:shb-sound-1}) [bend right] node [below, yshift=-0.3cm, xshift=-0.1cm]{\footnotesize{SHB}} to ({pic cs:shb-sound-2});
           \end{tikzpicture}
  \ea
 \eda

 Consider trace A.
 The HB relation orders critical sections based on their position
 in the trace and therefore $\unlockEE{y}{3} \hbSym \lockEE{y}{4}$
 where $\hbSym$ denotes the HB ordering relation.
 Hence, we find that $\writeEE{x}{1} \hbSym \writeEE{x}{5}$.
 This is a false negative. We are allowed to reorder the two critical sections
 and then two writes on~$x$ would appear right next to each other.
 Take $T' = [\thread{2}{\lockEE{y}{4}}, \thread{2}{\writeEE{x}{5}}, \thread{1}{\writeEE{x}{1}}]$
  where $T'$ represents an alternative schedule.

  Consider trace B.
  There are no critical sections.
  Hence, the conflicting events $\writeEE{x}{2}$ and $\writeEE{x}{5}$
  are unordered under the HB relation.
  This is a false positive.
  We assume that programs are executed under the sequential consistency
memory model~\citep{Adve:1996:SMC:619013.620590}.
  Hence, any reordering to exhibit the
  race among the writes on~$x$ violates the condition
  that each read must see the same (last) write.
  Consider the reordering
  $$T' = [\thread{2}{\writeEE{y}{1}}, \thread{2}{\readEE{y}{4}}, \thread{2}{\writeEE{x}{5}},
    \thread{1}{\writeEE{x}{2}}].$$
  In the original trace, the last write for $\readEE{y}{4}$ is $\writeEE{y}{3}$ but this
  does not apply to $T'$.
  As the read sees a different write, there is no guarantee that the events
  after the read on $y$ would take place.
 \end{example}

Mathur, Kini and Viswanathan~\citep{Mathur:2018:HFR:3288538:3276515}
show that the HB relation is only sound for the first race reported.
They introduce the schedulable happens-before (SHB) relation $\shbSym$.
The SHB relation additionally includes write-read dependencies
and therefore the two writes on~$x$ in the above trace B
are ordered under the SHB relation.
The SHB relation is sound in general but still incomplete
as critical sections are ordered by their position in the trace.

Kini, Mathur and Viswanathan~\citep{Kini:2017:DRP:3140587.3062374}
introduce the weak-causally precedes (WCP) relation.
Unlike HB and SHB, WCP reorders critical sections
under some conditions.
Recall trace A from Example~\ref{ex:hb-incomplete}.
Under WCP, events $\writeEE{x}{1}$ and $\writeEE{x}{5}$ are unordered.
Hence, WCP is more complete compared to HB and SHB.
Like HB, subsequent WCP races may be false positives.
See trace B in Example~\ref{ex:hb-unsound}
where $\writeEE{x}{2}$ and $\writeEE{x}{5}$ are not ordered under WCP
but this represents a false positive.

The WCP relation improves over the HB and SHB relation by being more complete.
However, WCP is still incomplete in general as shown by the following example.
\begin{example}
  \label{ex:wcp-incomplete}
  Consider
   \bda{lll}
   & \thread{1}{}
   & \thread{2}{}
   \\ \hline
   1. & \writeE{x} &
   \\ 2. & \lockE{y} &
   \\ 3. & \writeE{x} &
   \\ 4. & \unlockE{y} \tikzmark{wcp-incomplete-1} &
   \\ 5. && \lockE{y}
   \\ 6. && \tikzmark{wcp-incomplete-2} \writeE{x}
   \\ 7. && \unlockE{y}
   \eda
           \begin{tikzpicture}[overlay, remember picture, yshift=.25\baselineskip, shorten >=.5pt, shorten <=.5pt]
           \draw[->] ({pic cs:wcp-incomplete-1}) [bend right] node [below, yshift=-0.4cm, xshift=-0.3cm]{\footnotesize{WCP}} to ({pic cs:wcp-incomplete-2});
           \end{tikzpicture}
  Events $\writeEE{x}{1}$ and $\writeEE{x}{6}$ are in a predictable data race
  as witnessed by the following correctly reordered prefix
  $
   T' = [\lockEE{y}{5}, \writeEE{x}{1}, \writeEE{x}{6}].
  $
   WCP is unable to predict this race.

   The two critical sections contain conflicting events
   and therefore $\unlockEE{y}{4} \wcpSym \writeEE{x}{6}$.
   Then, we find that $\writeEE{x}{1} \wcpSym \writeEE{x}{6}$.
\end{example}

{\bf Lockset Method.}
A different method is based on the idea to compute
the set of locks that are held when processing a read/write event~\citep{Dinning:1991:DAA:127695:122767}.
We refer to this set as the \emph{lockset}.
For each event $e$ we compute its lockset $\LS{e}$
where $y \in \LS{e}$ if $e \in \csSym{y}$ for some critical section $\csSym{y}$.
Two conflicting events that are in a race
if their locksets are disjoint.

The computation of locksets is efficient and it is straightforward
to show that the lockset method is complete.
However, on its own the lockset method produces many false positives
as shown by our experiments later.

{\bf Hybrid Methods.}
The idea of Gen\c{c}, Roemer, Xu and Bond~\citep{10.1145/3360605}
is to pair up the lockset method with happens-before.
They introduce the strong-dependently precedes (SDP)
and the weak-dependently precedes (WDP) relation.
SDP and WDP are weaker compared to the earlier relations we have seen
where WDP is even weaker compared to SDP.
The lockset test is necessary to rule out (some) false positives.

Compared to WCP, SDP does not order critical sections if the conflicting events
are only writes and there is no read that follows the write in the later
critical section.
Hence, under SDP the two writes on~$x$ in Example~\ref{ex:wcp-incomplete} are unordered.
By weakening the WCP relation, the SDP relation on its own is no longer
strong enough to rule out false positives (in case of of the first reported race).

\begin{example}
\label{ex:sdp-needs-lockset}
  Consider
  \bda{lll}
   & \thread{1} & \thread{2}
  \\ \hline
  1. & \lockE{y} &
  \\ 2. & \writeE{x} &
  \\ 3. & \unlockE{y} &
  \\ 4. && \lockE{y}
  \\ 5. && \writeE{x}
  \\ 6. && \unlockE{y}
  \eda
  The two writes on $x$ are unordered under SDP but there is obviously no race
  as both writes are part of a critical section that involves the same lock~$y$.
  To deal with such cases, the SDP relation is paired with the lockset test.
\end{example}

SDP improves over WCP in case of write-write conflicting critical sections.
But as all other WCP conditions are still in place SDP remains incomplete.

\begin{example}
  \label{ex:r-w-cs}
  Consider
  \bda{lll}
  & \thread{1} & \thread{2}
  \\ \hline
  1.  & \writeE{x} &
  \\ 2. & \lockE{z} &
  \\ 3. & \readE{x} &
  \\ 4. & \writeE{y} &
  \\ 5. & \unlockE{z} \tikzmark{sdp-incomplete-1} &
  \\ 6. && \lockE{z}
  \\ 7. && \tikzmark{sdp-incomplete-2} \writeE{x}
  \\ 8. && \unlockE{z}
  \\ 9. && \writeE{y}
  \eda
           \begin{tikzpicture}[overlay, remember picture, yshift=.25\baselineskip, shorten >=.5pt, shorten <=.5pt]
           \draw[->] ({pic cs:sdp-incomplete-1}) [bend right] node [below, yshift=-0.4cm, xshift=-0.3cm]{\footnotesize{SDP}} to ({pic cs:sdp-incomplete-2});
           \end{tikzpicture}
   There is a read-write conflict on $x$ within the two critical sections.
   Hence, under SDP we find that
   $\writeEE{y}{4} \sdpSym \writeEE{y}{9}$.
   This is a false negative as there is a correct reordering under which both
   events appear right next to each other.
\end{example}

To achieve completeness,
\citet{10.1145/3360605} introduce the WDP relation.
WDP pretty much drops all of SDP's ordering conditions
among critical sections.
Two critical sections are ordered if
one contains a write and the other a conflicting read
where the write is the read's last write.

\begin{example}
\label{ex:wdp-must}
  Consider
  \bda{lll}
   & \thread{1} & \thread{2}
  \\ \hline
  1. & \writeE{z} &
  \\ 2. & \lockE{y} &
  \\ 3. & \writeE{x} &
  \\ 4. & \unlockE{y} \tikzmark{wdp-must-1} &
  \\ 5. && \lockE{y}
  \\ 6. && \tikzmark{wdp-must-2} \readE{x}
  \\ 7. && \unlockE{y}
  \\ 8. && \writeE{z}
  \eda
     \begin{tikzpicture}[overlay, remember picture, yshift=.25\baselineskip, shorten >=.5pt, shorten <=.5pt]
           \draw[->] ({pic cs:wdp-must-1}) [bend right] node [below, yshift=-0.4cm, xshift=-0.3cm]{\footnotesize{WDP}} to ({pic cs:wdp-must-2});
           \end{tikzpicture}
     We find that $\unlockEE{y}{4} \wdpSym \readEE{x}{6}$
     and therefore the two writes on~$z$ are ordered under WDP.
\end{example}

The WDP ordering condition among critical section is a necessary condition.
\citet{10.1145/3360605} show that for any predictable race the events
involved are unordered under WDP and their locksets are disjoint.
That is, the WDP relation combined with the lockset test is complete.

{\bf Our Work.}
We further strengthen the WDP relation while maintaining completeness.
Our approach is to strictly impose write-read dependencies (WRD) as
employed by the SHB relation in \citet{Mathur:2018:HFR:3288538:3276515}.
This allows us to filter out more false positives and also improves
the running time of the algorithm.

\begin{example}
\label{ex:pwr-is-stronger-than-wdp}
  Consider
  \bda{llll}
    & \thread{1} & \thread{2} & \thread{3} \\ \hline
1.  & \lockE{z} &&
\\ 2.  &  \writeE{y_1} \tikzmark{wrd-1-1} &&
\\ 3.  & \writeE{x} &&
\\ 4.  & \unlockE{z} \tikzmark{pwrcs-1} &&
\\ 5.  && \tikzmark{wrd-1-2} \readE{y_1} &
\\ 6.  && \writeE{y_2} \tikzmark{wrd-2-1} &
\\ 7.  &&&    \lockE{z}
\\ 8.  &&&    \tikzmark{wrd-2-2}  \readE{y_2} \tikzmark{pwrcs-2}
\\ 9. &&&     \unlockE{z}
\\ 10. &&&    \writeE{x}
\eda
           \begin{tikzpicture}[overlay, remember picture, yshift=.25\baselineskip, shorten >=.5pt, shorten <=.5pt]
           \draw[->] ({pic cs:wrd-1-1}) [bend right] node [below, yshift=-0.1cm, xshift=0.5cm]{\footnotesize{WRD}} to ({pic cs:wrd-1-2});
           \end{tikzpicture}
           \begin{tikzpicture}[overlay, remember picture, yshift=.25\baselineskip, shorten >=.5pt, shorten <=.5pt]
           \draw[->] ({pic cs:wrd-2-1}) [bend right] node [below, yshift=-0.4cm, xshift=-0.3cm]{\footnotesize{WRD}} to ({pic cs:wrd-2-2});
           \end{tikzpicture}
           \begin{tikzpicture}[overlay, remember picture, yshift=.25\baselineskip, shorten >=.5pt, shorten <=.5pt]
           \draw[->] ({pic cs:pwrcs-1}) [bend left=100]node [below, yshift=-0.4cm, xshift=2.7cm]{\footnotesize{PWR}} to ({pic cs:pwrcs-2});
           \end{tikzpicture}
WDP reports that the two writes on~$x$ are in a race.
This is a false positive.

Our \PWR\ relation includes the WRD relations
$\writeEE{y_1}{2} \wrdSym \readEE{y_1}{5}$
and $\writeEE{y_2}{6} \wrdSym \readEE{y_2}{8}$
and therefore $\unlockEE{z}{4} \pwrSym \readEE{y_2}{8}$.
\PWR\ stands for program order, write-read dependency order and ordered critical sections
(if events involved are ordered).
Hence, we find
that the two writes in $x$ are ordered under \PWR.
\PWR\ is stronger compared to WDP and therefore admits fewer false positives.
We can show that \PWR\ in combination with lockset is complete.
\end{example}

At the algorithmic level, \PWR\ has performance benefits as shown by the following example.

\begin{example}
  Consider
  \bda{llll}
  & \thread{1} & \thread{2} & \dots
  \\ \hline
 1. & \lockE{z} &&
 \\ 2. & \writeE{x_1} &&
 \\ 3. & \unlockE{z} \tikzmark{wdp-hist-1} &&
 \\ 4. & \writeE{y} \tikzmark{pwr-wrd-1} &&
 \\ 5. && \tikzmark{pwr-wrd-2} \readE{y} &
 \\ 6. && \lockE{z} &
 \\ 7. && \writeE{x_2} &
 \\ 8. && \unlockE{z} &
 \\ \dots &&&
 \\ 9.  && \lockE{z} &
 \\ 10. &&   \readE{x_1} \tikzmark{wdp-hist-2} &
 \\ 11. &&  \unlockE{z}
 \eda
            \begin{tikzpicture}[overlay, remember picture, yshift=.25\baselineskip, shorten >=.5pt, shorten <=.5pt]
           \draw[->] ({pic cs:wdp-hist-1}) [bend left=100] node [below, yshift=-0.1cm, xshift=1.5cm]{\footnotesize{WDP}} to ({pic cs:wdp-hist-2});
            \end{tikzpicture}
           \begin{tikzpicture}[overlay, remember picture, yshift=.25\baselineskip, shorten >=.5pt, shorten <=.5pt]
           \draw[->] ({pic cs:pwr-wrd-1}) [bend right] node [below, yshift=-0.3cm, xshift=-0.2cm]{\footnotesize{PWR}} to ({pic cs:pwr-wrd-2});
           \end{tikzpicture}

           To check if two critical sections are ordered,
           the algorithm that implements the WDP relation
needs to maintain a history of critical sections.
For each critical section, we record
(1) the writes for each variable, and
(2) the happens-before time for the release.
If there is a subsequent critical section (for the same lock)
with a read where the last write is in some earlier critical section,
then we need to enforce the WDP relation.
See $\unlockEE{z}{3} \wdpSym \readEE{x_1}{10}$.

The size of the history of critical sections as well as the writes per critical section
can be significantly large.
Our experiments show that this can have a significant impact on the performance.
\PWR\ improves over WDP as we do not maintain writes per critical section and can more
aggressively remove critical sections.

For our example,
due to the write-read dependency involving variable~$y$, the critical sections in thread 1 and 2 are ordered under \PWR.
Hence, thread 2 does not need to record thread 1's critical section at all.
Furthermore, we  only need to record to the happens-before time of the acquire instead
of all writes that are part of this critical section.
\end{example}

Another important contribution of our work is that we introduce a complete algorithm
that computes all predictable data race pairs.
The algorithm that implements the WDP relation is incomplete as shown by the following example.

\begin{example}
\label{ex:history-of-events}
  Consider
  \bda{lll}
  & \thread{1} & \thread{2}
  \\ \hline
 1. & \writeE{x} &
\\ 2. & \lockE{y} &
\\ 3. & \writeE{x} &
\\ 4. & \unlockE{y} &
\\ 5. && \lockE{y}
\\ 6. &&  \writeE{x}
\\ 7. &&  \unlockE{y}
\eda
There is a predictable race among $\writeEE{x}{1}$ and $\writeEE{x}{6}$.
Our algorithm that implements \PWR\ reports this race
but the algorithm that implements WDP, see Algorithm 2 in~\citet{10.1145/3360605},
does not report a race here.

The issue is that $\writeEE{x}{1}$ happens before $\writeEE{x}{3}$ (due to program order).
Algorithm 2 in~\citet{10.1145/3360605} only keeps the most 'recent' write per thread.
Hence, we have forgotten
about $\writeEE{x}{1}$ as we only kept $\writeEE{x}{3}$ by the time we reach $\writeEE{x}{6}$.
Events $\writeEE{x}{3}$ and $\writeEE{x}{6}$ are unordered under \PWR\ but they share a common lockset.
Hence, Algorithm 2 reports no race.

Our algorithm additionally records that $\writeEE{x}{1} \pwrSym \writeEE{x}{3}$.
Via $\writeEE{x}{3}$ we can derive that there is another potential race candidate $\writeEE{x}{1}$
that might be in a race with $\writeEE{x}{6}$.
Their locksets are disjoint and thus we report the race.
\end{example}

Maintaining $\writeEE{x}{1} \pwrSym \writeEE{x}{3}$ and identifying additional race candidates
requires extra time and space. Our algorithm requires a quadratic time and space.
We apply some optimizations under which
we obtain an efficient algorithm that runs in linear time and space.
The optimization may lead to incompleteness.
Our experiments show that this is mostly an issue in theory
but not for practical examples.

The upcoming section formalizes
the \PWR\ relation.
Section~\ref{sec:implementation} covers the implementation.
Experiments are presented in Section~\ref{sec:experiments}.

\section{The \PWR\ Relation}
\label{sec:shb-wcp-lockset}

We formally define the \PWR\ relation.

\begin{definition}[PO + WRD + ROD]
\label{def:www-relation}
  Let $T$ be a trace.
  We define a relation $\pwr{}{}$ among trace events
  as the smallest partial order that
  satisfies the following conditions:

  \begin{description}
  \item[Program order (PO):]
    Let $e, f \in T$ where
    $\compTID{e} = \compTID{f}$ and $\pos{e} < \pos{f}$.
    Then, we have that $\pwr{e}{f}$.
  \item[Write-read dependency (WRD):]
    Let $\writeEE{x}{j}, \readEE{x}{k} \in T$
    where $\writeEE{x}{j}$ is the last write of $\readEE{x}{k}$.
    That is,
        $j < k$ 
    and there is no other $\writeEE{x}{l}$
    such that $j < l < k$.
    Then, we have that $\pwr{\writeEE{x}{j}}{\readEE{x}{k}}$.
  \item[Release-order dependency (ROD):] Let $e,f \in T$ be two events.
    Let $\csSym{y}$, $\csSym{y}'$ be two critical sections
    where $e \in \csSym{y}$, $f \in \csSym{y}'$
     and $\pwr{e}{f}$.
    Then, we have that $\pwr{\rel{\csSym{y}}}{f}$.
  \end{description}

  We refer to $\pwr{}{}$ as the \emph{PO + WRD + ROD} (PWR)
  relation.
\end{definition}

We distinguish between write-write, read-write and write-read
race pair candidates. Write-write and read-write candidates
are not ordered under \PWR.
For write-read candidates we assume that the write is the last write
for the read under \PWR.

\begin{definition}[Lockset + PWR Write-Write and Read-Write Check]
  \label{def:lockset-w3-check}
  Let $T$ be a trace where
  $e, f$ are two conflicting events such that
  (1) $\LS{e} \cap \LS{f} = \emptyset$, (2) neither $e \pwrSym f$
  nor $f \pwrSym e$, and (3) $(e,f)$ is a write-write or read-write race pair.
  Then, we say that $(e,f)$ is a \emph{potential Lockset-PWR} data race pair.
\end{definition}

\begin{definition}[Lockset + PWR WRD Check]
  \label{def:w3-wrd-race-pairs}
  Let $T$ be a  trace.
  Let $e, f$ be two conflicting events
  such that $e$ is a write and $f$ a read
  where $\LS{e} \cap \LS{f} = \emptyset$,
  $e \pwrSym f$ and there is no $g$ such that
  $e \pwrSym g \pwrSym f$.
  Then, we say that $(e,f)$ is a \emph{potential Lockset-PWR WRD} data race pair.
\end{definition}

\begin{definition}[Potential Race Pairs via Lockset + PWR]
We write $\PotentialPWR{T}$ to denote the set of all potential Lockset-PWR (and WRD)
data race pairs as characterized by
Definitions~\ref{def:lockset-w3-check} and~\ref{def:w3-wrd-race-pairs}.
\end{definition}

\begin{proposition}[Lockset + PWR Completeness]
\label{prop:lockset-www-completeness}
Let $T$ be a  trace.
Let $e, f \in T$ such that $(e,f) \in \allPredRacesP{T}$.
Then, we find that $(e,f) \in \PotentialPWR{T}$.
\end{proposition}
Recall that $\allPredRacesP{T}$ denotes the set of
all predictable data race pairs (see the introduction).

We compare \PWR\ against WDP.

\begin{definition}[Weak-Dependently Precedes (WDP)~\citep{10.1145/3360605}]
  Let $T$ be a trace.
  We define a relation $\wdp{}{}$ among trace events
  as the smallest partial order that
  satisfies condition PO as well as the following conditions:

  \begin{description}
  \item[Weak Release-Conflict Dependency (RCD):] Let $e,f \in \rwTx$ be two conflicting events
    such that $f$ is a read event and $e$ is $f'$'s last write event.
    Let $\csSym{y}$, $\csSym{y}'$ be two critical sections
    where $f \in \csSym{y}$, $e \in \csSym{y}'$, $\pos{\rel{\csSym{y}}} < \pos{e}$.
    Then, $\wdp{\rel{\csSym{y}}}{e}$.
  \item[Release-Release Dependency (RRD):]
Let $e,f \in T$ be two events.
    Let $\csSym{y}$, $\csSym{y}'$ be two critical sections
    where $e \in \csSym{y}$, $f \in \csSym{y}'$
     and $\wdp{e}{f}$.
    Then, we have that $\wdp{\rel{\csSym{y}}}{\rel{\csSym{y'}}}$.
  \end{description}

  We refer to $\wdpSym$ as the \emph{weak-dependently precedes} (WDP)
  relation.~\footnote{Compared to the original WDP definition~\citep{10.1145/3360605}
  we do not distinguish between branch-dependent and branch-independent reads.
  We assume that all reads are branch-dependent. That is, each read may affect the control flow.
  This is more conservative but requires a much simpler tracing scheme
  where we do not have to inspect the program text.}
\end{definition}

Like \PWR, the WDP relation in combination with the lockset check is complete.
However, the \PWR\ relation is stronger and therefore allows us to rule
out more false positives.

\begin{proposition}[\PWR\ versus WDP]
\label{prop:pwr-vs-wdp}
  We have that $\wdpSym \subseteq \pwrSym$ but the reverse direction
  does not hold in general.
\end{proposition}

We can also state the that Lockset-\PWR\ check is sound under certain conditions.

\begin{proposition}[Lockset + \PWR\ Soundness for Two Threads]
\label{prop:lockset-www-two-threads-soundness}
  Let $T$ be a  trace that consists of at most two threads.
  Then, any potential Lockset-\PWR\ data race pair
  with an empty lockset is also a predictable data race pair.
\end{proposition}

Not every pair in $\PotentialPWR{T}$ is predictable.

\begin{example}
\label{ex:protected-by-wrds}
  Consider the following trace.
  \bda{lllll}
  & \thread{1}{}
  & \thread{2}{}
  & \thread{3}{}
  & \thread{4}{}
  \\ \hline
  1. & \lockE{y} &&&
  \\ 2. & \writeE{z_1} &&&
  \\ 3. && \readE{z_1} &&
  \\ 4. && \underline{\writeE{x}} &&
  \\ 5. && \writeE{z_2} &&
  \\ 6. & \readE{z_2} &&&
  \\ 7. & \unlockE{y} &&&
  \\ 8. &&& \lockE{y} &
  \\ 9. &&& \writeE{z_3} &
  \\ 10. &&&& \readE{z_3}
  \\ 11. &&&& \underline{\writeE{x}}
  \\ 12. &&&& \writeE{z_4}
  \\ 13. &&& \readE{z_4} &
  \\ 14. &&& \unlockE{y} &
  \eda
  Due to the write-read dependencies involving variables $z_1, z_2, z_3, z_4$,
  the two writes on $x$ are protected by the lock $y$.
  Hence, the pair $(\writeEE{x}{4}, \writeEE{x}{11})$
  is not a predictable data race pair.
  However, under \PWR\ events $\writeEE{x}{4}$, $\writeEE{x}{11})$ are unordered
  and their lockset is empty.
  Hence, the Lockset-\PWR\ method (falsely) reports the potential
  data race pair $(\writeEE{x}{4}, \writeEE{x}{11})$.
\end{example}

In the above example,
$\writeEE{x}{4}$ and $\writeEE{x}{11})$ is not the first potential race.
The WRDs on $z_1$, $z_2$, $z_3$ and $z_4$ are unprotected.
Hence, the first potential race involves $\writeEE{z_1}{2}$
and $\readEE{z_1}{3}$ (and this race is a predictable race).
However, each WRD can be protected via their own private lock.
Then, $\writeEE{x}{4}$ and $\writeEE{x}{11})$ becomes the first potential
race reported but this race is still a false positive.

Soundness is certainly an important property.
However, methods such as HB, WCP and SDP only guarantee
that the first race reported is sound but subsequent races may be false positives.
Algorithms/tools based on these methods commonly report as many (subsequent) races
as possible. In this light, we argue that the potential unsoundness
of the Lockset-\PWR\ check is not a serious practical issue.
The key advantage of \PWR\ is that we can reduce the number
of false positives compared to WDP.
The upcoming section shows how to compute $\PotentialPWR{T}$.
Our experiments show that our method works well in practice.



\section{The \PWREE\ Algorithm}
\label{sec:implementation}

Algorithm \PWREE\ computes $\PotentialPWR{T}$.
We start with an overview.

\subsection{Overview}
\label{sec:overview}

To implement the \PWR\ relation we combine ideas found in
FastTrack~\citep{flanagan2010fasttrack}, SHB~\citep{Mathur:2018:HFR:3288538:3276515}
and WCP~\citep{Kini:2017:DRP:3140587.3062374}.
For example, we employ vector clocks
and the more optimized epoch representation (FastTrack),
we manage a history of critical sections (WCP)
and track write-read dependencies (SHB).
Like the above algorithms, our algorithm also processes
events in a stream-based fashion and maintains
a set $\rwVC{x}$
of most recent reads/writes that are concurrent.
By concurrent we mean that the events are unordered under \PWR.
Elements in $\rwVC{x}$ are represented by their epoch
where each epoch allows us to uniquely identify the corresponding event.

Recall Example~\ref{ex:history-of-events} where we annotate the trace with $\rwVC{x}$.
For brevity, we omit vector clocks.
Instead of epochs, we write $w_i$ for a write at trace position $i$.
A similar notation is used for reads.

  \bda{llll}
  & \thread{1} & \thread{2} & \rwVC{x}
  \\ \hline
 1. & \writeE{x} & & \{ w_1 \}
\\ 2. & \lockE{y} &
\\ 3. & \writeE{x} & & \{ w_3 \}
\\ 4. & \unlockE{y} &
\\ 5. && \lockE{y}
\\ 6. &&  \writeE{x} & \{ w_3, w_6 \}
\\ 7. &&  \unlockE{y}
\eda
  We consider the various states of $\rwVC{x}$ while processing events.
  At trace position three, $w_3$ replaces $w_1$ due to the program order.
  At trace position six, we find $\rwVC{x} = \{ w_3, w_6 \}$.
  Under \PWR, $w_3$ and $w_6$ are concurrent but we do not report a race
  because their locksets share a common lock.

  The issue is that there is a race among $w_1$ and $w_6$
  but this race is not reported by standard
  single pass algorithms~\citep{flanagan2010fasttrack,Kini:2017:DRP:3140587.3062374,Mathur:2018:HFR:3288538:3276515,10.1145/3360605}.
  The reason is that $\rwVC{x}$ maintains only the most recent concurrent reads/writes.
  See the above example where $w_1$ is replaced by $w_3$.

  To cope with this issue we follow the \SHBEE{} two-pass
  algorithm~\citep{DBLP:conf/pppj/SulzmannS19}.
  In a first pass, we (1) maintain a history of replaced events
  and (2) reads/writes that are concurrent.
  The second pass traverses the history to discover all conflicting concurrent events.

  Here is our running example where this additional
  information has been annotated.
    \bda{llllll}
  & \thread{1} & \thread{2} & \rwVC{x} & \edges{x} & \concEvt{x}
  \\ \hline
 1. & \writeE{x} & & \{ w_1 \}
\\ 2. & \lockE{y} &
\\ 3. & \writeE{x} & & \{ w_3 \} & w_1 \gtEdge w_3
\\ 4. & \unlockE{y} &
\\ 5. && \lockE{y}
\\ 6. &&  \writeE{x} & \{ w_3, w_6 \} & & (w_3, w_6)
\\ 7. &&  \unlockE{y}
\eda
  The history is represented as a set $\edges{x}$ (E).
  Nodes connected via edges are reads/writes and can efficiently be represented
  via epochs (E). Concurrent events are stored in $\concEvt{x}$.
At trace position three, we record that $w_3$ replaces $w_1$.
At trace position six, we record that $w_3$ and $w_6$ are concurrent under \PWR.

Reporting of races is done in a second pass where we exploit
the information recorded in $\edges{x}$ and $\concEvt{x}$.
We consider all pairs in $\concEvt{x}$. If their locksets are disjoint
we report a race. This does not apply to the pair $(w_3, w_6)$, however,
this pair is crucial to discover further races.
From $(w_3, w_6)$ via $w_1 \gtEdge w_3$ we obtain a further
potential race candidate pair $(w_1, w_6)$.
Their locksets are disjoint and therefore we report the race pair $(w_1, w_6)$.
Thus, we are able to compute $\PotentialPWR{T}$.

The first pass enjoys the same time complexity as earlier
algorithms~\citep{flanagan2010fasttrack,Kini:2017:DRP:3140587.3062374,Mathur:2018:HFR:3288538:3276515,10.1145/3360605}.
The second pass comes with an additional quadratic run-time.
By limiting the size of elements in $\edges{x}$,
the second pass of traversing  $\edges{x}$ can be integrated into the first
pass where we build up $\edges{x}$.
This might lead to incompleteness but yields an efficient, linear run-time
algorithm. For practical examples it turns out that only maintaining
a maximum of 25 edge constraints at a time is a good compromise.

\begin{algorithm}
\caption{\PWREE{} algorithm (first pass)}\label{alg:w3poee-firstpass}

\begin{algorithmic}[1]
  \Function{w3}{$V, \LStSym$}
    \For {$y \in \LStSym$}
    \For {$(\thread{j}{k}, V') \in \LH{y}$}
     \If {$k < \accVC{V}{j}$}
     \State $V = V \sqcup V'$
     \EndIf
     \EndFor
     \EndFor

    \Return V
    \EndFunction
\end{algorithmic}

\begin{algorithmic}[1]
  \Procedure{acquire}{$i,y$}
  \State $\threadVC{i} = \Call{w3}{\threadVC{i},\LSt{i}}$
\State $\LSt{i} = \LSt{i} \cup \{y\}$
\State $\Acq{y} = \thread{i}{\accVC{\threadVC{i}}{i}}$
\State $\incC{\threadVC{i}}{i}$
\EndProcedure
\end{algorithmic}

\begin{algorithmic}[1]
  \Procedure{release}{$i,y$}
  \State $\threadVC{i} = \Call{w3}{\threadVC{i},\LSt{i}}$
\State $\LSt{i} = \LSt{i} - \{x\}$
\State $\LH{y} = \LH{y} \cup \{(\Acq{y}, \threadVC{i})\}$
\State $\incC{\threadVC{i}}{i}$
\EndProcedure
\end{algorithmic}

\begin{algorithmic}[1]
  \Procedure{write}{$i,x$}
  \State $\threadVC{i} = \Call{w3}{\threadVC{i},\LSt{i}}$
  \State $\vcEvt = \{ (\thread{i}{\accVC{\threadVC{i}}{i}}, \threadVC{i}, \LSt{i}) \} \cup \vcEvt$
\State $\edges{x} =
       \{ \thread{j}{k} \gtEdge \thread{i}{\accVC{\threadVC{i}}{i}}
       \mid \thread{j}{k} \in \rwVC{x} \wedge
       k < \accVC{\threadVC{i}}{j} \} \cup \edges{x}$
\State $\concEvt{x} = \{ (\thread{j}{k}, \thread{i}{\accVC{\threadVC{i}}{i}})
    \mid \thread{j}{k} \in \rwVC{x}
          \wedge k > \accVC{\threadVC{i}}{j} \} \cup \concEvt{x}$
  \State $\rwVC{x} = \{ \thread{i}{\accVC{\threadVC{i}}{i}} \}
         \cup \{ \thread{j}{k} \mid \thread{j}{k} \in \rwVC{x} \wedge
         k > \accVC{\threadVC{i}}{j} \}$
\State $\lastWriteVC{x} = \threadVC{i}$
\State $\lastWriteVCt{x} = i$
\State $\lastWriteVCL{x} = \LSt{i}$
\State $\incC{\threadVC{i}}{i}$
\EndProcedure
\end{algorithmic}

\begin{algorithmic}[1]
  \Procedure{read}{$i,x$}
      \State $j = \lastWriteVCt{x}$
  \If {$\accVC{\threadVC{i}}{j} < \accVC{\lastWriteVC{x}}{j} \wedge \LSt{i} \cap \lastWriteVCL{x} = \emptyset$}
    \State $reportPotentialRace(\thread{i}{\accVC{\threadVC{i}}{i}}, \thread{j}{\accVC{\lastWriteVC{x}}{j}})$
  \EndIf
\State $\threadVC{i} = \threadVC{i} \sqcup \lastWriteVC{x}$
\State $\threadVC{i} = \Call{w3}{\threadVC{i},\LSt{i}}$
\State $\vcEvt = \{ (\thread{i}{\accVC{\threadVC{i}}{i}}, \threadVC{i}, \LSt{i}) \} \cup \vcEvt$
\State $\edges{x} =
       \{ \thread{j}{k} \gtEdge \thread{i}{\accVC{\threadVC{i}}{i}}
       \mid \thread{j}{k} \in \rwVC{x} \wedge
       k < \accVC{\threadVC{i}}{j} \} \cup \edges{x}$
\State $\concEvt{x} = \{ (\thread{j}{k}, \thread{i}{\accVC{\threadVC{i}}{i}})
    \mid \thread{j}{k} \in \rwVC{x}
          \wedge k > \accVC{\threadVC{i}}{j} \} \cup \concEvt{x}$
\State $\rwVC{x} = \{ \thread{i}{\accVC{\threadVC{i}}{i}} \}
         \cup \{ \thread{j}{k} \mid \thread{j}{k} \in \rwVC{x} \wedge
         k > \accVC{\threadVC{i}}{j} \}$
\State $\incC{\threadVC{i}}{i}$
\EndProcedure
\end{algorithmic}

\end{algorithm}

\subsection{First Pass}

Algorithm~\ref{alg:w3poee-firstpass} specifies the first pass of \PWREE{}
and computes $\edges{x}$ and $\concEvt{x}$.
Events are processed in a stream-based fashion.
For each event we find a procedure that deals with this event.
We immediately report write-read races. Reporting of write-write and read-write
races takes place in a second pass.

We compute the lockset for read/write events
and check if read/write events are concurrent by establishing the \PWR\ relation.
To check if events are in \PWR\ relation we make use of vector clocks and epochs.
We first define vector clocks and epochs and introduce various
state variables maintained by the algorithm that rely on these concepts.

For each thread~$i$ we compute the current set $\LSt{i}$ of locks held by this thread.
We use $\LSt{i}$ to avoid confusion with the earlier introduced
set $\LS{e}$ that represents the lockset for event $e$.
We have that $\LS{e} = \LSt{i}$ where $\LSt{i}$ is the set at the time we process event $e$.
Initially, $\LSt{i} = \emptyset$ for all threads~$i$.

The algorithm also maintains several vector clocks.

\begin{definition}[Vector Clocks]
  A \emph{vector clock} $V$ is a list of \emph{time stamps} of the following form.
  \bda{rcl}
   V  & ::= & [i_1,\dots,i_n]
   \eda
   We assume vector clocks are of a fixed size $n$.
   Time stamps are natural numbers and each time stamp position $j$ corresponds to the thread
   with identifier $j$.

   We define
   $$\supVC{[i_1,\dots, i_n]}{[j_1,\dots,j_n]} \ \ = \ \ [\maxN{i_1}{j_1},\dots,\maxN{i_n}{j_n}]$$
   to synchronize two vector clocks by building the point-wise maximum.

 We write $\accVC{V}{j}$ to access the time stamp at position $j$.
 We write $\incC{V}{j}$ as a short-hand for incrementing the vector clock $V$ at position $j$ by one.

 We define vector clock $V_1$ to be smaller than vector clock $V_2$,
written $V_1 < V_2$,
if (1) for each thread $i$, $i$'s time stamp in $V_1$ is smaller or equal
compared to $i$'s time stamp in $V_2$, and
(2) there exists a thread $i$ where $i$'s time stamp in $V_1$
is strictly smaller compared to $i$'s time stamp in $V_2$.
\end{definition}
If the vector clock assigned to event $e$ is smaller
compared to the vector clock assigned to $f$,
then we can argue that $e$ happens before $f$.
For $V_1 = \supVC{V_2}{V_3}$ we find that $V_1 \leq V_2$ and $V_1 \leq V_3$.

For each thread $i$ we maintain a vector clock $\threadVC{i}$.
For each shared variable $x$ we find vector clock
$\lastWriteVC{x}$ to maintain the last write access on $x$.
Initially, for each vector clock $\threadVC{i}$
all time stamps are set to 0 but position $i$ where the time stamp is set to 1.
For $\lastWriteVC{x}$ all time stamps are set to~0.

To efficiently record read and write events we make use of epochs~\citep{flanagan2010fasttrack}.
\begin{definition}[Epoch]
  Let $j$ be a thread id and $k$ be a time stamp.
  Then, we write $\thread{j}{k}$ to denote an \emph{epoch}.
\end{definition}
Each event $e$ can be uniquely associated to an epoch~$\thread{j}{k}$.
Take its vector clock and extract the time stamp $k$ for the thread $j$
the event~$e$ belongs to. For each event this pair of information
represents a unique key to locate the event.
Hence, we sometimes abuse notation and write $e$ when referring
to the epoch of event~$e$.

Via epochs we can also check if events are in a happens-before relation
without having to take into account the events vector clocks.

\begin{proposition}[FastTrack~\citep{flanagan2010fasttrack} Epochs]
\label{prop:epoch-conc-check-main}
  Let $T$ be some trace.
  Let $e, f$ be two events in $T$ where (1) $e$ appears before $f$ in $T$,
  (2) $e$ is in thread $j$, and (3) $f$ is in thread $i$.
  Let $V_1$ be $e$'s vector clock and $V_2$ be $f$'s vector clock
  computed by the FastTrack algorithm.
  Then, we have that $e$ and $f$ are concurrent w.r.t.~the $\hbSym$ relation iff
  $\accVC{V_2}{j} < \accVC{V_1}{j}$.
\end{proposition}
HB-concurrent holds when comparing vector clocks $V_2 < V_1$.
If $\accVC{V_2}{j} < \accVC{V_1}{j}$ then the vector clocks of thread $j$
and $i$ have not been synchronized. Therefore, $e$ and $f$ must be concurrent.
Similar argument applies for the direction from right to left.
Our algorithm is an extension of FastTrack. Hence, the above property
carries over to our algorithm and the \PWR\ relation.

For each lock variable $y$, we find $\Acq{y}$ to record the last
entry point to the critical section guarded by lock $y$.
$\Acq{y}$ is represented by an epoch.
The set $\LH{y}$ maintains the lock history for lock variables $y$.
For each critical section we record the pair $(\Acq{y}, V)$
where $\Acq{y}$ is the acquire's epoch and $V$ is the vector clock of the corresponding release event.
We refer to $(\Acq{y}, V)$ as a \emph{lock history element} for a critical section represented
by a matching acquire/release pair.
Based on the information recorded in $\LH{y}$ we are able to efficiently
apply the ROD rule as we will see shortly.
The set $\LH{y}$ is initially empty. The initial definition of $\Acq{y}$ can be left unspecified
as by the time we access $\Acq{y}$, $\Acq{y}$ has been set.

For each shared variable $x$, the set $\rwVC{x}$ maintains
the current set of concurrent read/write events.
Each event is represented  the event's epoch.
The set $\rwVC{x}$ is initially empty.

The first-pass of \PWREE{} maintains three further sets that are important
during the second pass. All sets are initially empty.

The set $\edges{x}$ keeps track of the events from $\rwVC{x}$
that will be replaced when processing reads/writes.
If $e$ replaces $f$ this means that $e$ happens-before $f$ w.r.t.~\PWR.
We record this information by adding the \emph{edge constraint} $f \gtEdge e$.

The set $\concEvt{x}$ keeps track for each variable $x$
of the set of potential race pairs  where the events involved are concurrent to each
other w.r.t.~\PWR.
Such pairs represent potential write-write and read-write pairs.
We \emph{do not} enforce that their locksets must be disjoint because via
a pair $(e,f) \in \concEvt{x}$ where $e, f$ share a common lock
we may be able to reach a concurrent pair $(g,f)$
where the locksets of $g$ and $f$ are disjoint.
Recall the example from Section~\ref{sec:overview}.
For convenience, for all race pairs $(e,f)$ collected
by $\concEvt{x}$ we maintain the property that $\pos{e} < \pos{f}$.
For write-write pairs this property always holds.
For read-write pairs the read is usually put first. Strictly following the trace position
order makes the second pass easier to formalize as we will see shortly.

The set $\vcEvt$ records for each read/write event its lockset and vector clock
at the time of processing.
We add the triple consisting of the event's epoch, lockset and vector clock
to the set $\vcEvt$. The epoch serves as unique key for lookup.
The information stored $\vcEvt$ in will be used during the second pass.
We traverse chains of edge constraints starting from candidates
in $\concEvt{x}$ to build new candidates. Each such found candidate must
satisfy the Lockset-\PWR\ check (see Definition~\ref{def:lockset-w3-check}).
Based on the information stored in $\vcEvt$ we can carry out
this check easily.

Finally, we make use of $\lastWriteVCt{x}$ to record the thread id of the last write
and $\lastWriteVCL{x}$ to record the last write's lockset.
This information in combination with $\lastWriteVC{x}$ is used to check
for potential write-read race pairs.

In summary, the first pass of \PWREE{} maintains the following (global) variables:
\begin{itemize}
\item $\LSt{i}$, set of locks held by thread~$i$.
\item $\threadVC{i}$, vector clock for thread~$i$.
\item $\lastWriteVC{x}$, vector clock of last write on $x$.
\item $\lastWriteVCt{x}$, thread id of last write on $x$.
\item $\lastWriteVCL{x}$, lockset of last write on $x$.
\item $\rwVC{x}$, current set of concurrent reads/writes on $x$.
\item $\Acq{y}$, epoch of last acquire on $y$.
\item $\LH{y}$, lock history for $y$.
\item $\lastWriteVC{x}$, last write access for $x$.
\item $\edges{x}$, set of edge constraints for $x$.
\item $\concEvt{x}$, accumulated set of concurrent reads/writes on $x$.
\item $\vcEvt$,  set of lockset and vector clock for each read/write.
\end{itemize}

We have now everything in place to consider the various cases
covered by the first pass of \PWREE{}.

For each event we need to establish the \PWR\ relation.
In particular, we need to apply the ROD rule from Definition~\ref{def:www-relation}.
Establishing the ROD rule is done via helper function \Call{w3}{}.

Instead of some event $e \in \csSym{y}$ as formulated
in the ROD rule, it suffices to consider
the acquire event of $\csSym{y}$.
In Appendix~\ref{sec:w3-variants} we show that this is indeed sufficient.

The slightly revised ROD rule then reads as follows.
If $\csSym{y}$ appears before $\csSym{y}'$ in the trace, $f \in \csSym{y}'$
and $\pwr{\acq{\csSym{y}}}{f}$,
then $\pwr{\rel{\csSym{y}}}{f}$.
Event $f$ is represented by the two parameters $V$ and $\LStSym$.
V is $f's$ vector clock and $\LStSym$ is the set of locks held when processing $f$.
For each $y \in \LStSym$ we check all prior critical sections on the same
lock in the lock history $\LH{y}$.
Each element is represented as a pair $(\thread{j}{k}, V')$
where $\thread{j}{k}$ is the epoch of the acquire and $V'$ the vector clock
of the matching release.
The check $k < \accVC{V}{j}$ tests if the acquire happens-before $f$, i.e.~$\pwra{\acq{\csSym{y}}}{f}$. \PWR\ then demands that $\pwr{\rel{\csSym{y}}}{f}$.
This is guaranteed by $V = V \sqcup V'$.

In case of an acquire event in thread~i on lock variable $y$,
we first apply the ROD rule via helper function \Call{w3}{}.
Then, we extend the thread's lockset
with $y$. In $\Acq{y}$ we record the epoch of the acquire event.
Finally, we increment the thread's time stamp to indicate that the event has been processed.

When processing the corresponding release event,
we again apply first the ROD rule.
Then, we remove $y$ from the thread's lockset.
We add the pair $(\Acq{y}, \threadVC{i})$ to $\LH{y}$.
$\LH{y}$ accumulates the \emph{complete} lock history.
There is no harm doing so but this can be of course inefficient.
Optimizations to remove lock history elements are discussed
later.

Next, we consider processing of write events.
We apply first the ROD rule.
Then we add the event's information to $\vcEvt$.
We update $\concEvt{x}$ by checking if the write is concurrent to any of the events in $\rwVC{x}$.
As discussed above, there is no need to compare vector clocks to check if two events are
concurrent to each other. It suffices to compare epochs.
Similarly, we update $\rwVC{x}$ but only maintain the current set of concurrent reads/writes.
Finally, we update the ``last write'' information and increment the thread's time stamp.

We consider processing of read events.
We first check for a potential write-read race pair
by checking if the read is concurrent to the last write and their locksets are disjoint.
If the check is successful we immediately report the pair.
Only after this check we impose the write-read dependency by synchronizing
the last writes vector clock with the vector clock of the current thread.
Then, we call $\Call{w3}{}$ to apply the ROD rule.
Updates for $\vcEvt$, $\concEvt{x}$ and $\rwVC{x}$ are the same as in case of write.

\subsection{Second Pass}

The first pass yields the set $\edges{x}$ of edge constraints
and the set $\concEvt{x}$ of read/write pairs that are concurrent
under \PWR. In a second pass, we compute further concurrent pairs by
systematically traversing $\edges{x}$ starting with elements
from $\concEvt{x}$.
The thus obtained pairs are collected in some set $\accConcs{x}$.
Computation of $\accConcs{x}$ is defined as follows.

\begin{definition}[\PWREE\ Reporting Race Candidates]
  \label{def:all-concs-post}
  Let $\concEvt{x}$ and $\edges{x}$ be obtained by \PWREE\
  for all shared variables $x$.

  We define a total order among pairs in $\concEvt{x}$ as follows.
  Let $(e, f) \in \concEvt{x}$ and $(e', f') \in \concEvt{x}$.
  Then, we define $(e,f) < (e',f')$ if $\pos{e} < \pos{e'}$.

  For each variable $x$, we compute the set $\accConcs{x}$
  by repeatedly performing the following steps. Initially,
  $\accConcs{x} = \{ \}$.

  \begin{enumerate}
  \item If $\concEvt{x} = \{\}$ stop.
  \item Otherwise, let $( e, f)$ be the smallest element in $\concEvt{x}$.
  \item Let $G=\{g_1,\dots, g_n\}$ be maximal such that
    $g_1 \gtEdge e, \dots, g_n \gtEdge e \in \edges{x}$ and
    $\pos{g_1} < \dots < \pos{g_n}$.
  \item $\accConcs{x} := \{ (e, f) \} \cup \accConcs{x}$.
  \item $\concEvt{x} := \{ (g_1, f), \dots, (g_n,f) \} \cup (\concEvt{x} - \{ (e,f) \})$.
  \item Repeat.
\end{enumerate}

\end{definition}

We can state that the set $\accConcs{x}$ covers all concurrent reads/writes on $x$.

\begin{proposition}
 \label{prop:w3po-post}
 Let $T$ be a  trace of size $n$ and $x$ be some shared variable.
  Let $\allConcsP{T}{x} = \{ (e,f) \mid e,f \in \rwTx \wedge \pos{e} < \pos{f} \wedge
  e \not\pwrSym f \wedge f \not\pwrSym e \}$
  Let $x$ be a variable.
  Then, construction of $\accConcs{x}$ takes time $O(n*n)$
  and $\allConcsP{T}{x} \subseteq \accConcs{x}$.
\end{proposition}
We assume that the number of distinct (shared) variables $x$ is a constant.
Hence, construction of all sets $\accConcs{x}$ takes time $O(n*n)$.

For each pair in $\accConcs{x}$ we yet need to carry out the lockset check.
We can retrieve the lockset for each event by consulting the set $\vcEvt$.
The set $\vcEvt$ records for each read/write event
its lockset and vector clock at the time of processing.
The vector clock is needed because
the set $\accConcs{x}$ overapproximates the set of concurrent reads/writes.
Hence, we not only need to filter out pairs that share a common lock
but also pairs that are not concurrent.

\begin{example}
 \label{ex:w3po-post}
  Consider the following trace annotated with $\rwVC{x}$, $\edges{x}$ and $\concEvt{x}$.
  We omit explicit vector clocks and epochs for brevity
  and write $w_i$ ($r_i$) for a write (read) at trace position $i$.

  \bda{lllllll}
  & \thread{1} & \thread{2} & \thread{3} & \rwVC{x} & \edges{x} & \concEvt{x}
  \\ \hline
  1. & \writeE{x} &&& \{ w_1 \} &&
  \\ 2. & \writeE{y_1} &&& \{ w_1 \} &&
  \\ 3. && \readE{y_1} && \{ w_1 \} &&
  \\ 4. && \writeE{y_2} &&\{ w_1 \} &&
  \\ 5. && \writeE{x} && \{ w_5 \} & w_1 \gtEdge w_5 &
  \\ 6. &&& \readE{y_2} & \{ w_5 \} &&
  \\ 7. &&& \writeE{x} & \{ w_5, w_7 \} && (w_5,w_7)
  \eda
  Besides writes on $x$, we also find reads/writes on variables $y_1$ and $y_2$.
  We do not keep track of these events as their sole purpose is to enforce
  via some write-read dependencies that $w_1 \pwrSym w_7$.

  $\PWREE$ yields $\concEvt{x} = \{ (w_5, w_7) \}$ and $\edges{x} = \{ w_1 \gtEdge w_5 \}$.
  The second pass then yields $\accConcs{x} = \{ (w_5, w_7), (w_1, w_7) \}$.
  However, $(w_1, w_7) \not\in\allConcsP{T}{x}$  because $w_1 \pwrSym w_7$.
\end{example}

The example shows that the set $\accConcs{x}$ may contain some non-concurrent pairs.
To filter out such pairs we apply the concurrency test
specified in Proposition~\ref{prop:epoch-conc-check-main}.
We consult the vector clock of the event appearing later in the trace
and check if the time stamp of the event appear first in the trace is greater.
We also check that locksets are disjoint.

\begin{lemma}[Lockset + \PWR\ Filtering]
 \label{le:lockset-w3-filtering}
  Let $x$ be some variable.
  Let $\vcEvt$ be obtained by $\PWREE$
  and $\accConcs{x}$ via $\PWREE$'s second pass.
  Let $(\thread{i}{k}, \thread{j}{l}) \in \accConcs{x}$
  and $(\thread{j}{l}, L_2, V_2) \in \vcEvt$.
  If  $L_1 \cap L_2 = \emptyset$ and  $k > \accVC{V_2}{j}$
  then $(\thread{i}{k}, \thread{j}{l})$ is either a write-write or read-write pair
  in $\PotentialPWR{T}$ where we use the event's epoch as a unique identifier.
\end{lemma}

We conclude that $\PWREE$ (first pass) yields all write-read pairs
in $\PotentialPWR{T}$ and the second pass followed by filtering
yields all write-write and read-write pairs in $\PotentialPWR{T}$.

\subsection{Time and Space Complexity}

We consider the time and space complexity of $\PWREE$ including
first, second pass and filtering.
Let $n$ be the size of trace $T$, $k$ be the number of threads
and $c$ be the number of critical sections.
We consider the number of variables  as a constant.

We first consider the time complexity of $\PWREE$ (first pass).
The size of the vector clocks and the set $\rwVC{x}$ is bounded by $O(k)$.
Each processing step of $\PWREE$ requires adjustments of a constant
number of vector clocks. This takes time $O(k)$.
Adjustment of sets $\concEvt{x}$, $\edges{x}$ and $\rwVC{x}$
requires to consider $O(k)$ epochs where each
comparison among epochs is constant. Altogether, this requires time $O(k)$.
We consider $\vcEvt$ as a map where adding a new element takes constant time.
The Lockset-\PWR\ WRD race check takes constant time as we assume lookup
of time stamp is constant and the size of each lockset is a constant.
Each call to \Call{w3}{} takes time $O(c)$.
Overall, $\PWREE$ takes time $O(n*k + n*c)$ to process trace $T$.

The space required by $\PWREE$ is as follows.
Sets $\vcEvt$, $\concEvt{x}$ and $\edges{x}$ require $O(n*k)$ space.
This applies to $\vcEvt$ because
for each event the size of the vector clock is $O(k)$.
The size of the lockset is assumed to be a constant.
Each element in $\concEvt{x}$ and $\edges{x}$ requires
constant space. In each step, we may add $O(k)$ new elements
because the size of $\rwVC{x}$ is bounded by $O(k)$.
Set $\LH{y}$ requires space $O(c*k)$.
Overall, $\PWREE$ requires $O(n*k + c*k)$ space.

The time for the second pass is $O(n*n)$. There are $O(n*n)$ pairs where each pair requires
constant space. Hence, $O(n*n)$ space is required.
Filtering for each candidate takes constant time.
The size of the lockset is constant, time stamp comparison is a constant
and lookup of locksets and vector clocks in $\vcEvt$ is assumed to take constant time.

Overall, the run-time of $\PWREE$, first and second pass, including filtering
is $O(n*k + n*c + n*n)$. The space requirement is  $O(n*k + c*k + n*n)$.
Parameters $k$ and $c$ are bounded by $O(n)$.
Hence, the run-time of $\PWREE$ is $O(n*n)$.

\subsection{Optimizations}

There are a number optimizations, e.g.~aggressive filtering
and removal of critical sections, that can be carried.
Details are discussed in Appendix~\ref{sec:wwwpoee-optimizations}.
These optimizations will not change the theoretical time complexity but are essential
in a practical implementation.

We can turn $\PWREE$ into a single-pass, linear run-time algorithm
if we impose a limit on the history of critical sections
and a limit on the number of edge constraints.
Then, we can merge the second pass into the first pass.
We refer to this variant as \PWREELimit{}.

Imposing a limit on the number of edge constraints
means that the second pass (traversal of edges) and filtering
takes place during the first pass as well.
Whenever candidates are added to $\concEvt{x}$ we immediately
apply the steps described in Definition~\ref{def:all-concs-post} (but
the number of edge constraints to consider is limited)
and carry out the filtering check.

By imposing a limit on the number of edge constraints in $\edges{x}$, we might miss out
on some potential data race pairs. For example, consider the case
of 27 subsequent writes in one thread followed by a write in another thread.
We assume that each write is connected to a distinct code location.
In our implementation, we treat events connected to the same code location as the same event.
Each of the 27 subsequent writes is in a race with the write in the other thread.
There are 27 race pairs overall but a standard single-pass algorithm
would only report the \emph{last} race pair.
The 27 subsequent writes give rise to 26 edge constraints.
As we only maintain 25 edge constraints, we fail to report the \emph{first} data race.
In our experience, limiting the size of $\edges{x}$ to 25 turns out to be a good compromise.

Consider the history of critical sections $\LH{y}$.
Instead of a global history, our implementation maintains
thread-local histories.
The number of thread-local histories (after applying optimizations)
is only bounded by the number of threads and the number of distinct variables.
This can still be a fairly high number and requires extra management effort.
In our implementation, we simply impose a fixed limit on the
size of thread-local histories. If the limit is exceeded,
the newly added element simply overwrites the oldest element.
This might have the consequence that  two events may become unordered
w.r.t.~the \emph{limited} \PWR\ relation (where they should be ordered without limit).
Completeness is unaffected but our method may produce more false positives.
In our experience, limiting the size of thread-local histories
to five turns out to be a good compromise.

\section{Experiments}
\label{sec:experiments}

\begin{table}
      \caption{Benchmark results. The time is given in minutes:seconds, maximum memory consumption in megabytes.}
    \label{tab:realworldbench-singlephase}
    {\footnotesize
    \begin{tabular}{l|l|l|l|l|l|l|l}
    \textbf{} & FT &  $\SHBEELimit{}$ & WCP &  TSan & $\PWREELimit{}$ & \PWRZero{}  & \PWREE \\ \hline
    \textbf{Avrora}  &  \textbf{}& \textbf{} &&&&\\
    Races:  & 20 & 20(0) &  & 30 & 20(0) & 20 & 20(0) \\
    Time: & 0:14 & 0:19 & $>$30  & 0:22 & 0:21 & 0:17 & 0:22 \\
    Mem:  & 2125 & 3965 & 6385 & 2934 & 3886  & 1999 & 4048 \\
    \hline
    \textbf{Batik}  & \textbf{}&  \textbf{} &&&& \\
    Races:  & 12 & 4(0) & 12 &  12 & 4(0) & 4 & 4(0)  \\
    Time:   & 0:01 & 0:01 & 0:02 &  0:01 & 0:01 & 0:01 & 0:01 \\
    Mem:   & 29 & 35 & 84 & 33 & 68 & 32  & 80 \\
    \hline
    \textbf{H2} & \textbf{}& \textbf{} &&&&\\
    Races:  & 125 & 248(0) &  & 672 & 252(2) & 252 &   \\
    Time: & 1:35 & 2:22 & $>30$ & 4:52 & 2:48  & 1:55  & 3:56 \\
    Mem:  & 2154 & 13431 & 6350 &  4998 & 16393 & 3465  & > 32gb \\
    \hline
    \textbf{Lusearch} & \textbf{}&  \textbf{}&&&&\\
    Races: & 15 & 15(0) & 15 & 19 & 15(0) &  15  & 15(0) \\
    Time: &0:01 & 0:02 & 0:19 & 0:01 & 0:04  & 0:04 & 0:04\\
    Mem: & 14 & 14 & 8685 & 11 & 1848  & 1852  & 2243 \\
    \hline
    \textbf{Tomcat} & \textbf{}&  \textbf{}&&&&\\
    Races: & 636 & 681(194) &  & 1984 & 823(219)  & 623  & 823(219) \\
    Time: & 0:33 & 0:49 & $>$30 & 0:37 & 0:51 & 0:36 & 23:19 \\
    Mem: & 12245 & 13617 & 13268 & 7523 & 19919 & 14861  & 28452 \\
    \hline
    \textbf{Xalan}& \textbf{}& \textbf{}&&&& \\
    Races: & 41 & 44(0) & 142 & 244 & 394(223) & 185  & \\
    Time: & 1:19 & 2:04 & 7:11 & 1:33 & 2:30 & 1:30  & 1:51 \\
    Mem: & 7282 & 9591 & 14882 &  5342 & 24980  & 7284  & > 32gb \\ \hline
    \textbf{Moldyn} & \textbf{}&  \textbf{}&&&&\\
    Races: & 33& 24(8) & 33 & 56 & 24(8)  & 18  & 24(8) \\
    Time: & 0:32& 0:54 & 0:37 & 0:46 & 0:55  & 0:33  & 1:23 \\
    Mem: & 99 & 487 & 108 & 91 & 515 & 71 & 19833 \\ \hline
    \end{tabular}
    }
    \end{table}


{\bf Test Candidates and Benchmarks.}
The test candidates
are FastTrack(FT), \SHBEELimit{}, WCP, ThreadSanitizer (TSan),
  \PWRZero{}, \PWREELimit{} and \PWREE.
\SHBEELimit{} and \PWREELimit{} limit the size of edge constrains to 25. The limit for histories is five.
\PWRZero{} is a variant of \PWREELimit{} where the limit for edge constraints is zero.
\PWREE\ does not impose any limits and therefore requires two passes whereas all the other candidates
run in a single pass.
We have implemented all of them in a common framework for better comparability.

We have not implemented SDP and WDP.
As WCP, SDP and WDP rely on effectively the same method, their performance
in terms of time and space is similar.
See Table 8 in~\citet{10.1145/3360605}
where running times and space usage of WCP, SDP and WDP are compared.
Hence, only including WCP
allows for a fair comparison.

In terms of precision,
the WDP algorithm has more false positives and false negatives compared to \PWREELimit{} and \PWREE.
\citep{10.1145/3360605} make use of an additional Vindication phase
to check for a witness to confirm that a reported race is not a false positive.
Vindication requires extra time and there is no guarantee to filter out all false positives as Vindication is incomplete
(if no witness is found after one try Vindication gives up).

We have carried out experiments that involve two benchmark suites.
The first benchmark suite consists of test of the Java Grande benchmark suite~\citep{smith2001parallel} and the DaCapo benchmark suite (version 9.12, \citep{Blackburn:2006:DBJ:1167473.1167488}). This is a standard set of real-world tests to measure the performance in terms of execution time and memory consumption.
The second benchmark suite consists of small, tricky examples
found in earlier works~\citet{Mathur:2018:HFR:3288538:3276515,
  Roemer:2018:HUS:3296979.3192385, 10.1145/3371085, roemer2019practical,Kini:2017:DRP:3140587.3062374} and our own examples that we found while working with different race prediction algorithms.
For these examples we know the exact number of predictable races
and therefore we can measure
the precision (false positives, false negatives) of our test candidates.
In terms of precision, \PWREELimit{} performs
the best among all test candidates. The limits employed by \PWREELimit{} yields the same results as for \PWREE.
We refer to Appendix~\ref{sec:precision} for details.

{\bf Performance.}
For benchmarking we use an AMD Ryzen 7 3700X and 32 gb of RAM with Ubuntu 18.04 as operating system.
We have evaluated the performance of all benchmarks from the Java Grande and DaCapo benchmark suites. For space reasons, we only discuss the results for some benchmarks.
Other benchmarks not discussed have similar characteristics
compared to the ones show in Table \ref{tab:realworldbench-singlephase}.
The time is given in minutes and seconds (mm:ss). The memory consumption is also measured for the complete program and not only for the single algorithms. In row \textit{Mem} the memory consumption is given in megabytes. We use the standard `time' program in Ubuntu to measure the time and memory consumption.

For TSan, \PWREELimit{}, \PWRZero{} and \SHBEELimit{}, entry row $\#Races$
shows the number of reported data race pairs. We filter pairs connected
to the same code locations. For \PWREELimit{} and \SHBEELimit{} we write 24(8) if 24 data race pairs were reported which includes 8 that were found using edge constraints.
Because we only count pairs connected to unique code locations, it is possible that
\PWRZero{} reports a pair that will be reported via edge constraints in case of \PWREELimit{} and \PWREE.
Appendix~\ref{sec:counting-data-races} explains this point in more detail.
For FastTrack(FT) and WCP the number of races are the number of data race connected to distinct code locations.

In terms of number of races reported, \PWREELimit{} performs
best followed by \PWRZero{} and \SHBEELimit{}.
\SHBEELimit{} only reports races from trace-specific schedules.
\PWRZero{} lacks edge constraints which leads to missed races
as shown by the test cases Tomcat, Xalan and Moldyn.
FastTrack, TSan and WCP report significantly fewer races.

FastTrack has the best performance in terms of run-time and memory consumption.
TSan also shows good run-time performance with the exception of the H2 test case.
The reason is due to our use of vector clocks in our TSan implementation.

    WCP has performance problems with the Avrora, H2 and Tomcat test cases.
    For all three cases we aborted the experiment after 30 minutes.
    The reason are several thousands of critical sections that seem to be checked for each read/write inside a critical section.
    Like \PWREELimit{}, WCP maintains a history of critical sections but (a) needs to track more information (all read/write accesses within a critical section), and
    (b) can remove critical section not as aggressively as
    \PWREELimit{} because write-read dependencies are not strictly enforced.
    As argued above, similar observations should apply to SDP and WCP.

    Table 8 in~\citet{10.1145/3360605} shows reasonable performance for WCP, SDP and WDP
for Avrora, H2 and Tomcat. We are a bit surprised here but the base time (first column
in Table 8) seems to indicate that in our measurements the programs were running for much
longer and then some performance issues seem to arise.
Our time measurements do not show the base time nor the time it took to generate
the trace. We only show the timings to carry out the analysis.

\PWREE{} has no performance problems for test cases Avrora and Tomcat.
This shows that our approach of dealing with the history of critical sections
appears to be superior in comparison to WCP and WDP.
For H2 and Xalan, \PWREE{} runs out of memory.
The timings indicate the point in time when out of memory occurred.
The H2 test case consists of 100 million events and the Xalan test
case consists of 80 million events.
This leads a huge number of edge constraints which then leads to out of memory.
Hence, limiting the number of edge constraints is crucial for
achieving an acceptable performance.

The memory consumption of \PWREELimit{} is still high for some
test cases such as H2, Tomcat and Xalan.
Overall, the run-times of \PWREELimit{} are competitive
compared to the fastest candidate FastTrack.
In summary, \PWREELimit{} strives for a balance between
good performance and high precision.

\section{Related Work}
\label{sec:related-work}

We review further works in the area of dynamic data race prediction.

{\bf Efficient methods.}
We have already covered the efficient (linear-time)
data race prediction methods that found
use in FastTrack~\citep{flanagan2010fasttrack}, SHB~\citep{Mathur:2018:HFR:3288538:3276515},
WCP~\citep{Kini:2017:DRP:3140587.3062374},
SDP/WDP~\citep{10.1145/3360605} and TSan~\citep{serebryany2009threadsanitizer}.
TSan is also sometimes referred to as ThreadSanitizer v1.

The newer TSan version, ThreadSanitizer v2 (TSanV2)~\citep{thread-sanitizer-v2},
is an optimized version of the FastTrack algorithm in terms of performance.
TSanV2 only keeps a limited history of write/read events.
This improves the performance but results in a higher number of false negatives.

Acculock \citep{xie2013acculock}  optimizes the original TSan algorithm
by employing a single lockset per variable.
Acculock can be faster, but is less precise compared to TSan
if a thread uses multiple locks at once.


SimpleLock \citep{yu2016simplelock+} uses a simplified lockset algorithm.
A data race is only reported if at least one of the accesses is not protected by any lock.
They show that they are faster compared to Acculock but miss more data races
since they do not predict data races for events with different locks.

%
%
%
%
%

{\bf Semi-efficient methods.}
We consider semi-efficient methods that require polynomial run-time.

The \SHBEE{} algorithm~\citep{DBLP:conf/pppj/SulzmannS19} requires
quadratic run-time to compute all trace-specific data race pairs.
Our \PWREE{} algorithm adopts ideas from \SHBEE{}
and achieves completeness while retaining a quadratic run-time.
By limiting the history of edge constraints, the variant \PWREELimit{} runs in linear time.
Due to this optimization we are only near complete.
In practice, the performance gain outweighs the benefit of a higher precision.

The Vindicator algorithm \citep{Roemer:2018:HUS:3296979.3192385} improves the WCP algorithm
and is sound for all reported data races. It can predict more data races compared to WCP,
but requires three phases to do so. The first phase of Vindicator is a weakened WCP relation
that removes the happens-before closure. For the second phase, it constructs a graph that
contains all events from the processed trace. This phase is unsound and incomplete
which is why a third phase is required. The third phase makes a single attempt
to reconstruct a witness trace for the potential data race and reports a data race if successful.
Vindicator has a much higher run-time compared to the ``PWR'' family of algorithms.
We did not include Vindicator in our measurements as we experienced performance issues
for a number of real world benchmarks (e.g.~timeout due to lack of memory etc).

The M2 algorithm~\citep{10.1145/3371085} can be seen as a further improvement of the Vindicator idea.
Like Vindicator, multiple phases are required. M2 requires two phases.
M2 has $O(n^4)$ run-time (where $n$ is the size of the trace).
M2 is sound and like \PWREE{} complete for two threads.
The measurements by~\citet{10.1145/3371085} show that in terms of precision
M2 improves over FastTrack, SHB, WCP
and Vindicator for a subset of the real-world benchmarks that we also considered.
We did not include M2 in our measurements as we are not aware of any publicly
available implementation.

{\bf Exhaustive methods.}
We consider methods that are sound \emph{and} complete to which we refer as exhaustive methods.
Exhaustive methods come with a high degree of precision but generally are no longer efficient.

The works by~\citet{DBLP:conf/rv/SerbanutaCR12, Huang:2014:MSP:2666356.2594315, LuoHuangRosuSystematic15} use SAT/SMT-solvers to derive alternative feasible traces from a recorded trace. These traces can be checked with an arbitrary race prediction algorithm for data races. This requires multiple phases and is rather complimentary to the algorithms that we compare in this work as any of them could be used to check the derived traces for data races.

Kalhauge and Palsberg~\citep{kalhauge2018sound} present a data race prediction algorithm that is sound and complete. They also use an SMT-solver to derive alternative feasible traces. The algorithm inspects write-read dependencies in more detail, to determine at which point the control flow might be influenced by the observed write-read dependency. Deriving multiple traces and analyzing their write-read dependencies for their influence on the control flow is a very slow process that can take several hours according to their benchmarks.

{\bf Comparative studies.}
Previous works that compare multiple data race prediction algorithms use the Java Grande~\citep{smith2001parallel}, Da Capo \citep{Blackburn:2006:DBJ:1167473.1167488} and IBM Contest \citep{farchi2003concurrent} benchmark suits to do so.
The DaCapo and Java Grande benchmark suite contain real world programs with an unknown amount of data races and other errors.
The IBM Contest benchmark is a set of very small programs with known concurrency bugs like data races.

Yu, Park, Chun and Bae~\citep{yu2017experimental} compare the performance of FastTrack \citep{flanagan2010fasttrack}, SimpleLock+ \citep{yu2016simplelock+},
Multilock-HB, Acculock \citep{xie2013acculock} and Casually-Precedes (CP) \citep{Smaragdakis:2012:SPR:2103621.2103702} with a subset of the benchmarks
found in the DaCapo, JavaGrande and IBM Contest suits. They reimplemented CP to use a sliding window of only 1000 shared memory events which does not
affect the soundness but the amount of predicted data races. In our work we compare newer algorithms including Weak-Casually-Precedes which is the successor of~CP.

The work by \citet{liao2017dataracebench} compares Helgrind, ThreadSanitizer Version 2, Archer and the Intel Inspector.
They focus on programs that make use of OpenMP for parallelization. OpenMP uses synchronization primitives that are unknown
to Helgrind, ThreadSanitizer v2 and the Intel Inspector. Only the Archer race predictors is optimized for OpenMP.
For their comparison they use the Linpack and SPECOMP benchmark suits for which the number of concurrency errors is unknown.
Most of their races are enforced by including OpenMP primitives to parallelize the code which are not part of the original
implementation. Thus, they lack complex concurrency patterns. In some related work~\citep{lin2018runtime}, the same authors test the four data race
predictors from their previous work again with programs that make use of OpenMP and SIMD parallelism. Since SIMD is unsupported
by all tested data race predictors, they encounter a high number of false positives. The data race predictors we tested would
report many false positives for the same reasons.

The work by \citet{alowibdi2013empirical} evaluates the static data race predictors RaceFuzzer, RacerAJ, Jchord, RCC and JavaRace\-Finder.
They only evaluate the performance and the number of data races that each algorithms predicts. Static data race prediction
is known to report too many false positives since they need to over-approximate the program behavior. We only tested dynamic
data race predictors that make use of a recorded trace to predict data races. In terms of accuracy we expect that our test
candidates perform better compared to the static data race predictors.

Yu, Yang, Su and Ma~\citep{yu2017evaluation} test Eraser, Djit+, Helgrind+, ThreadSanitizer v1, FastTrack, Loft, Acculock, Multilock-HB,
Simplelock and Simplelock+. It is the to the best of our knowledge the only previous work that includes ThreadSanitizer v1.
In their work, they use the original implementations for testing. They test the performance and accuracy with the unit tests of
ThreadSanitizer. The tested data race predictors ignore
write-read dependency and are therefore only sound for the first predicted data race. We test current solutions that mostly include
write-read dependencies. For accuracy testing we included a set of handwritten test cases to
ensure that every algorithm sees the same order of events. All algorithms, except Vindicator, are reimplemented in a common framework to
ensure that all algorithms use the same utilities and have the same parsing overhead.


\section{Conclusion}
\label{sec:conclusion}

We have introduced \PWREE{} and the practically inspired variant \PWREELimit{}.
\PWREELimit{} is an efficient, near complete and often sound dynamic data race prediction algorithm
that combines the lockset method with recent improvements
made in the area of happens-before based methods.
\PWREE{} is complete in theory.
For the case of two threads we can show that \PWREE{} is also sound.
Our experimental results show that \PWREELimit{} performs
well compared to the state-of-the art efficient data race prediction algorithms.
The implementation of \PWREELimit{} including all contenders
as well as benchmarks can be found at
\begin{center}
  \url{https://github.com/KaiSta/SpeedyGo}.~\footnote{The ``PWR'' algorithm
    is referred to as ``W3PO'' in the SpeedyGo framework.}
\end{center}



\pagebreak

\bibliography{main}


\begin{thebibliography}{30}


\ifx \showCODEN    \undefined \def \showCODEN     #1{\unskip}     \fi
\ifx \showDOI      \undefined \def \showDOI       #1{#1}\fi
\ifx \showISBNx    \undefined \def \showISBNx     #1{\unskip}     \fi
\ifx \showISBNxiii \undefined \def \showISBNxiii  #1{\unskip}     \fi
\ifx \showISSN     \undefined \def \showISSN      #1{\unskip}     \fi
\ifx \showLCCN     \undefined \def \showLCCN      #1{\unskip}     \fi
\ifx \shownote     \undefined \def \shownote      #1{#1}          \fi
\ifx \showarticletitle \undefined \def \showarticletitle #1{#1}   \fi
\ifx \showURL      \undefined \def \showURL       {\relax}        \fi
\providecommand\bibfield[2]{#2}
\providecommand\bibinfo[2]{#2}
\providecommand\natexlab[1]{#1}
\providecommand\showeprint[2][]{arXiv:#2}

\bibitem[\protect\citeauthoryear{Adve and Gharachorloo}{Adve and
  Gharachorloo}{1996}]%
        {Adve:1996:SMC:619013.620590}
\bibfield{author}{\bibinfo{person}{Sarita~V. Adve} {and}
  \bibinfo{person}{Kourosh Gharachorloo}.} \bibinfo{year}{1996}\natexlab{}.
\newblock \showarticletitle{Shared Memory Consistency Models: A Tutorial}.
\newblock \bibinfo{journal}{{\em Computer\/}} \bibinfo{volume}{29},
  \bibinfo{number}{12} (\bibinfo{date}{Dec.} \bibinfo{year}{1996}),
  \bibinfo{pages}{66--76}.
\newblock
\showISSN{0018-9162}
\showDOI{%
\url{https://doi.org/10.1109/2.546611}}


\bibitem[\protect\citeauthoryear{Alowibdi and Stenneth}{Alowibdi and
  Stenneth}{2013}]%
        {alowibdi2013empirical}
\bibfield{author}{\bibinfo{person}{Jalal~S Alowibdi} {and}
  \bibinfo{person}{Leon Stenneth}.} \bibinfo{year}{2013}\natexlab{}.
\newblock \showarticletitle{An empirical study of data race detector tools}. In
  \bibinfo{booktitle}{{\em 2013 25th Chinese Control and Decision Conference
  (CCDC)}}. IEEE, \bibinfo{pages}{3951--3955}.
\newblock
\showDOI{%
\url{https://doi.org/10.1109/CCDC.2013.6561640}}


\bibitem[\protect\citeauthoryear{Blackburn, Garner, Hoffmann, Khang, McKinley,
  Bentzur, Diwan, Feinberg, Frampton, Guyer, Hirzel, Hosking, Jump, Lee, Moss,
  Phansalkar, Stefanovi\'{c}, VanDrunen, von Dincklage, and
  Wiedermann}{Blackburn et~al\mbox{.}}{2006}]%
        {Blackburn:2006:DBJ:1167473.1167488}
\bibfield{author}{\bibinfo{person}{Stephen~M. Blackburn},
  \bibinfo{person}{Robin Garner}, \bibinfo{person}{Chris Hoffmann},
  \bibinfo{person}{Asjad~M. Khang}, \bibinfo{person}{Kathryn~S. McKinley},
  \bibinfo{person}{Rotem Bentzur}, \bibinfo{person}{Amer Diwan},
  \bibinfo{person}{Daniel Feinberg}, \bibinfo{person}{Daniel Frampton},
  \bibinfo{person}{Samuel~Z. Guyer}, \bibinfo{person}{Martin Hirzel},
  \bibinfo{person}{Antony Hosking}, \bibinfo{person}{Maria Jump},
  \bibinfo{person}{Han Lee}, \bibinfo{person}{J.~Eliot~B. Moss},
  \bibinfo{person}{Aashish Phansalkar}, \bibinfo{person}{Darko Stefanovi\'{c}},
  \bibinfo{person}{Thomas VanDrunen}, \bibinfo{person}{Daniel von Dincklage},
  {and} \bibinfo{person}{Ben Wiedermann}.} \bibinfo{year}{2006}\natexlab{}.
\newblock \showarticletitle{The {DaCapo} Benchmarks: Java Benchmarking
  Development and Analysis}. In \bibinfo{booktitle}{{\em Proc.\ of OOPSLA
  '06}}. \bibinfo{publisher}{ACM}, \bibinfo{pages}{169--190}.
\newblock
\showDOI{%
\url{https://doi.org/10.1145/1167515.1167488}}


\bibitem[\protect\citeauthoryear{Dinning and Schonberg}{Dinning and
  Schonberg}{1991}]%
        {Dinning:1991:DAA:127695:122767}
\bibfield{author}{\bibinfo{person}{Anne Dinning} {and} \bibinfo{person}{Edith
  Schonberg}.} \bibinfo{year}{1991}\natexlab{}.
\newblock \showarticletitle{Detecting Access Anomalies in Programs with
  Critical Sections}.
\newblock \bibinfo{journal}{{\em SIGPLAN Not.\/}} \bibinfo{volume}{26},
  \bibinfo{number}{12} (\bibinfo{date}{Dec.} \bibinfo{year}{1991}),
  \bibinfo{pages}{85--96}.
\newblock
\showISSN{0362-1340}
\showDOI{%
\url{https://doi.org/10.1145/127695.122767}}


\bibitem[\protect\citeauthoryear{Farchi, Nir, and Ur}{Farchi
  et~al\mbox{.}}{2003}]%
        {farchi2003concurrent}
\bibfield{author}{\bibinfo{person}{Eitan Farchi}, \bibinfo{person}{Yarden Nir},
  {and} \bibinfo{person}{Shmuel Ur}.} \bibinfo{year}{2003}\natexlab{}.
\newblock \showarticletitle{Concurrent bug patterns and how to test them}. In
  \bibinfo{booktitle}{{\em Proceedings International Parallel and Distributed
  Processing Symposium}}. IEEE, \bibinfo{pages}{7--pp}.
\newblock
\showDOI{%
\url{https://doi.org/10.1109/IPDPS.2003.1213511}}


\bibitem[\protect\citeauthoryear{Fidge}{Fidge}{1992}]%
        {Fidge:1991:PAT:646210.683620}
\bibfield{author}{\bibinfo{person}{Colin~J. Fidge}.}
  \bibinfo{year}{1992}\natexlab{}.
\newblock \showarticletitle{Process Algebra Traces Augmented with Causal
  Relationships}. In \bibinfo{booktitle}{{\em Proc.\ of FORTE '91}}.
  \bibinfo{publisher}{North-Holland Publishing Co.},
  \bibinfo{address}{Amsterdam, The Netherlands, The Netherlands},
  \bibinfo{pages}{527--541}.
\newblock


\bibitem[\protect\citeauthoryear{Flanagan and Freund}{Flanagan and
  Freund}{2010}]%
        {flanagan2010fasttrack}
\bibfield{author}{\bibinfo{person}{Cormac Flanagan} {and}
  \bibinfo{person}{Stephen~N Freund}.} \bibinfo{year}{2010}\natexlab{}.
\newblock \showarticletitle{{FastTrack}: efficient and precise dynamic race
  detection}.
\newblock \bibinfo{journal}{{\it Commun. ACM}} \bibinfo{volume}{53},
  \bibinfo{number}{11} (\bibinfo{year}{2010}), \bibinfo{pages}{93--101}.
\newblock
\showDOI{%
\url{https://doi.org/10.1145/1543135.1542490}}


\bibitem[\protect\citeauthoryear{Gen\c{c}, Roemer, Xu, and Bond}{Gen\c{c}
  et~al\mbox{.}}{2019}]%
        {10.1145/3360605}
\bibfield{author}{\bibinfo{person}{Kaan Gen\c{c}}, \bibinfo{person}{Jake
  Roemer}, \bibinfo{person}{Yufan Xu}, {and} \bibinfo{person}{Michael~D.
  Bond}.} \bibinfo{year}{2019}\natexlab{}.
\newblock \showarticletitle{Dependence-Aware, Unbounded Sound Predictive Race
  Detection}.
\newblock \bibinfo{journal}{{\em Proc. ACM Program. Lang.\/}}
  \bibinfo{volume}{3}, \bibinfo{number}{OOPSLA}, Article
  \bibinfo{articleno}{179} (\bibinfo{date}{Oct.} \bibinfo{year}{2019}),
  \bibinfo{numpages}{30}~pages.
\newblock
\showDOI{%
\url{https://doi.org/10.1145/3360605}}


\bibitem[\protect\citeauthoryear{Huang, Meredith, and Rosu}{Huang
  et~al\mbox{.}}{2014}]%
        {Huang:2014:MSP:2666356.2594315}
\bibfield{author}{\bibinfo{person}{Jeff Huang}, \bibinfo{person}{Patrick~O'Neil
  Meredith}, {and} \bibinfo{person}{Grigore Rosu}.}
  \bibinfo{year}{2014}\natexlab{}.
\newblock \showarticletitle{Maximal Sound Predictive Race Detection with
  Control Flow Abstraction}.
\newblock \bibinfo{journal}{{\em SIGPLAN Not.\/}} \bibinfo{volume}{49},
  \bibinfo{number}{6} (\bibinfo{date}{June} \bibinfo{year}{2014}),
  \bibinfo{pages}{337--348}.
\newblock
\showISSN{0362-1340}
\showDOI{%
\url{https://doi.org/10.1145/2666356.2594315}}


\bibitem[\protect\citeauthoryear{Kalhauge and Palsberg}{Kalhauge and
  Palsberg}{2018}]%
        {kalhauge2018sound}
\bibfield{author}{\bibinfo{person}{Christian~Gram Kalhauge} {and}
  \bibinfo{person}{Jens Palsberg}.} \bibinfo{year}{2018}\natexlab{}.
\newblock \showarticletitle{Sound Deadlock Prediction}.
\newblock \bibinfo{journal}{{\em Proc. ACM Program. Lang.\/}}
  \bibinfo{volume}{2}, \bibinfo{number}{OOPSLA}, Article
  \bibinfo{articleno}{146} (\bibinfo{date}{Oct.} \bibinfo{year}{2018}),
  \bibinfo{numpages}{29}~pages.
\newblock
\showISSN{2475-1421}
\showDOI{%
\url{https://doi.org/10.1145/3276516}}


\bibitem[\protect\citeauthoryear{Kini, Mathur, and Viswanathan}{Kini
  et~al\mbox{.}}{2017}]%
        {Kini:2017:DRP:3140587.3062374}
\bibfield{author}{\bibinfo{person}{Dileep Kini}, \bibinfo{person}{Umang
  Mathur}, {and} \bibinfo{person}{Mahesh Viswanathan}.}
  \bibinfo{year}{2017}\natexlab{}.
\newblock \showarticletitle{Dynamic Race Prediction in Linear Time}.
\newblock \bibinfo{journal}{{\em SIGPLAN Not.\/}} \bibinfo{volume}{52},
  \bibinfo{number}{6} (\bibinfo{date}{June} \bibinfo{year}{2017}),
  \bibinfo{pages}{157--170}.
\newblock
\showDOI{%
\url{https://doi.org/10.1145/3062341.3062374}}


\bibitem[\protect\citeauthoryear{Lamport}{Lamport}{1978}]%
        {lamport1978time}
\bibfield{author}{\bibinfo{person}{Leslie Lamport}.}
  \bibinfo{year}{1978}\natexlab{}.
\newblock \showarticletitle{Time, clocks, and the ordering of events in a
  distributed system}.
\newblock \bibinfo{journal}{{\it Commun. ACM}} \bibinfo{volume}{21},
  \bibinfo{number}{7} (\bibinfo{year}{1978}), \bibinfo{pages}{558--565}.
\newblock
\showDOI{%
\url{https://doi.org/10.1145/359545.359563}}


\bibitem[\protect\citeauthoryear{Liao, Lin, Asplund, Schordan, and Karlin}{Liao
  et~al\mbox{.}}{2017}]%
        {liao2017dataracebench}
\bibfield{author}{\bibinfo{person}{Chunhua Liao}, \bibinfo{person}{Pei-Hung
  Lin}, \bibinfo{person}{Joshua Asplund}, \bibinfo{person}{Markus Schordan},
  {and} \bibinfo{person}{Ian Karlin}.} \bibinfo{year}{2017}\natexlab{}.
\newblock \showarticletitle{DataRaceBench: a benchmark suite for systematic
  evaluation of data race detection tools}. In \bibinfo{booktitle}{{\em
  Proceedings of the International Conference for High Performance Computing,
  Networking, Storage and Analysis}}. ACM, \bibinfo{pages}{11}.
\newblock
\showDOI{%
\url{https://doi.org/10.1145/3126908.3126958}}


\bibitem[\protect\citeauthoryear{Lin, Liao, Schordan, and Karlin}{Lin
  et~al\mbox{.}}{2018}]%
        {lin2018runtime}
\bibfield{author}{\bibinfo{person}{Pei-Hung Lin}, \bibinfo{person}{Chunhua
  Liao}, \bibinfo{person}{Markus Schordan}, {and} \bibinfo{person}{Ian
  Karlin}.} \bibinfo{year}{2018}\natexlab{}.
\newblock \showarticletitle{Runtime and memory evaluation of data race
  detection tools}. In \bibinfo{booktitle}{{\em International Symposium on
  Leveraging Applications of Formal Methods}}. Springer,
  \bibinfo{pages}{179--196}.
\newblock
\showDOI{%
\url{https://doi.org/10.1007/978-3-030-03421-4_13}}


\bibitem[\protect\citeauthoryear{Luo, Huang, and Rosu}{Luo
  et~al\mbox{.}}{2015}]%
        {LuoHuangRosuSystematic15}
\bibfield{author}{\bibinfo{person}{Qingzhou Luo}, \bibinfo{person}{Jeff Huang},
  {and} \bibinfo{person}{Grigore Rosu}.} \bibinfo{year}{2015}\natexlab{}.
\newblock \bibinfo{booktitle}{{\em Systematic Concurrency Testing with Maximal
  Causality}}.
\newblock \bibinfo{type}{{T}echnical {R}eport}.
\newblock


\bibitem[\protect\citeauthoryear{Mathur, Kini, and Viswanathan}{Mathur
  et~al\mbox{.}}{2018}]%
        {Mathur:2018:HFR:3288538:3276515}
\bibfield{author}{\bibinfo{person}{Umang Mathur}, \bibinfo{person}{Dileep
  Kini}, {and} \bibinfo{person}{Mahesh Viswanathan}.}
  \bibinfo{year}{2018}\natexlab{}.
\newblock \showarticletitle{What Happens-after the First Race? Enhancing the
  Predictive Power of Happens-before Based Dynamic Race Detection}.
\newblock \bibinfo{journal}{{\em Proc. ACM Program. Lang.\/}}
  \bibinfo{volume}{2}, \bibinfo{number}{OOPSLA}, Article
  \bibinfo{articleno}{145} (\bibinfo{date}{Oct.} \bibinfo{year}{2018}),
  \bibinfo{numpages}{29}~pages.
\newblock
\showISSN{2475-1421}
\showDOI{%
\url{https://doi.org/10.1145/3276515}}


\bibitem[\protect\citeauthoryear{Mattern}{Mattern}{1989}]%
        {Mattern89virtualtime}
\bibfield{author}{\bibinfo{person}{Friedemann Mattern}.}
  \bibinfo{year}{1989}\natexlab{}.
\newblock \showarticletitle{Virtual Time and Global States of Distributed
  Systems}. In \bibinfo{booktitle}{{\em Parallel and Distributed Algorithms}}.
  \bibinfo{publisher}{North-Holland}, \bibinfo{pages}{215--226}.
\newblock


\bibitem[\protect\citeauthoryear{Pavlogiannis}{Pavlogiannis}{2019}]%
        {10.1145/3371085}
\bibfield{author}{\bibinfo{person}{Andreas Pavlogiannis}.}
  \bibinfo{year}{2019}\natexlab{}.
\newblock \showarticletitle{Fast, Sound, and Effectively Complete Dynamic Race
  Prediction}.
\newblock \bibinfo{journal}{{\em Proc. ACM Program. Lang.\/}}
  \bibinfo{volume}{4}, \bibinfo{number}{POPL}, Article
  \bibinfo{articleno}{Article 17} (\bibinfo{date}{Dec.} \bibinfo{year}{2019}),
  \bibinfo{numpages}{29}~pages.
\newblock
\showDOI{%
\url{https://doi.org/10.1145/3371085}}


\bibitem[\protect\citeauthoryear{Roemer, Gen{\c{c}}, and Bond}{Roemer
  et~al\mbox{.}}{2019}]%
        {roemer2019practical}
\bibfield{author}{\bibinfo{person}{Jake Roemer}, \bibinfo{person}{Kaan
  Gen{\c{c}}}, {and} \bibinfo{person}{Michael~D Bond}.}
  \bibinfo{year}{2019}\natexlab{}.
\newblock \showarticletitle{Practical Predictive Race Detection}.
\newblock \bibinfo{journal}{{\em arXiv preprint arXiv:1905.00494\/}}
  (\bibinfo{year}{2019}).
\newblock


\bibitem[\protect\citeauthoryear{Roemer, Gen\c{c}, and Bond}{Roemer
  et~al\mbox{.}}{2018}]%
        {Roemer:2018:HUS:3296979.3192385}
\bibfield{author}{\bibinfo{person}{Jake Roemer}, \bibinfo{person}{Kaan
  Gen\c{c}}, {and} \bibinfo{person}{Michael~D. Bond}.}
  \bibinfo{year}{2018}\natexlab{}.
\newblock \showarticletitle{High-coverage, Unbounded Sound Predictive Race
  Detection}.
\newblock \bibinfo{journal}{{\em SIGPLAN Not.\/}} \bibinfo{volume}{53},
  \bibinfo{number}{4} (\bibinfo{date}{June} \bibinfo{year}{2018}),
  \bibinfo{pages}{374--389}.
\newblock
\showISSN{0362-1340}
\showDOI{%
\url{https://doi.org/10.1145/3192366.3192385}}


\bibitem[\protect\citeauthoryear{Serbanuta, Chen, and Rosu}{Serbanuta
  et~al\mbox{.}}{2012}]%
        {DBLP:conf/rv/SerbanutaCR12}
\bibfield{author}{\bibinfo{person}{Traian{-}Florin Serbanuta},
  \bibinfo{person}{Feng Chen}, {and} \bibinfo{person}{Grigore Rosu}.}
  \bibinfo{year}{2012}\natexlab{}.
\newblock \showarticletitle{Maximal Causal Models for Sequentially Consistent
  Systems}. In \bibinfo{booktitle}{{\em Poc.\ of RV'12}} {\em
  (\bibinfo{series}{LNCS})}, Vol.~\bibinfo{volume}{7687}.
  \bibinfo{publisher}{Springer}, \bibinfo{pages}{136--150}.
\newblock
\showDOI{%
\url{https://doi.org/10.1007/978-3-642-35632-2_16}}


\bibitem[\protect\citeauthoryear{Serebryany and Iskhodzhanov}{Serebryany and
  Iskhodzhanov}{2009}]%
        {serebryany2009threadsanitizer}
\bibfield{author}{\bibinfo{person}{Konstantin Serebryany} {and}
  \bibinfo{person}{Timur Iskhodzhanov}.} \bibinfo{year}{2009}\natexlab{}.
\newblock \showarticletitle{ThreadSanitizer: data race detection in practice}.
  In \bibinfo{booktitle}{{\em Proc.\ of WBIA ’09}}. \bibinfo{publisher}{ACM},
  \bibinfo{address}{New York, NY, USA}, \bibinfo{pages}{62--71}.
\newblock
\showDOI{%
\url{https://doi.org/10.1145/1791194.1791203}}


\bibitem[\protect\citeauthoryear{Smaragdakis, Evans, Sadowski, Yi, and
  Flanagan}{Smaragdakis et~al\mbox{.}}{2012}]%
        {Smaragdakis:2012:SPR:2103621.2103702}
\bibfield{author}{\bibinfo{person}{Yannis Smaragdakis}, \bibinfo{person}{Jacob
  Evans}, \bibinfo{person}{Caitlin Sadowski}, \bibinfo{person}{Jaeheon Yi},
  {and} \bibinfo{person}{Cormac Flanagan}.} \bibinfo{year}{2012}\natexlab{}.
\newblock \showarticletitle{Sound Predictive Race Detection in Polynomial
  Time}.
\newblock \bibinfo{journal}{{\em SIGPLAN Not.\/}} \bibinfo{volume}{47},
  \bibinfo{number}{1} (\bibinfo{date}{Jan.} \bibinfo{year}{2012}),
  \bibinfo{pages}{387--400}.
\newblock
\showISSN{0362-1340}
\showDOI{%
\url{https://doi.org/10.1145/2103656.2103702}}


\bibitem[\protect\citeauthoryear{Smith, Bull, and Obdrizalek}{Smith
  et~al\mbox{.}}{2001}]%
        {smith2001parallel}
\bibfield{author}{\bibinfo{person}{Lorna~A Smith}, \bibinfo{person}{J~Mark
  Bull}, {and} \bibinfo{person}{J Obdrizalek}.}
  \bibinfo{year}{2001}\natexlab{}.
\newblock \showarticletitle{A Parallel {Java} Grande Benchmark Suite}. In
  \bibinfo{booktitle}{{\em Proc.\ of SC'01}}. IEEE, \bibinfo{pages}{8--8}.
\newblock
\showDOI{%
\url{https://doi.org/10.1145/582034.582042}}


\bibitem[\protect\citeauthoryear{Sulzmann and Stadtm{\"{u}}ller}{Sulzmann and
  Stadtm{\"{u}}ller}{2019}]%
        {DBLP:conf/pppj/SulzmannS19}
\bibfield{author}{\bibinfo{person}{Martin Sulzmann} {and} \bibinfo{person}{Kai
  Stadtm{\"{u}}ller}.} \bibinfo{year}{2019}\natexlab{}.
\newblock \showarticletitle{Predicting all data race pairs for a specific
  schedule}. In \bibinfo{booktitle}{{\em Proc.\ of {MPLR}'19}}.
  \bibinfo{publisher}{{ACM}}, \bibinfo{address}{New York, NY, USA},
  \bibinfo{pages}{72--84}.
\newblock
\showDOI{%
\url{https://doi.org/10.1145/3357390.3361022}}


\bibitem[\protect\citeauthoryear{ThreadSanitizer}{ThreadSanitizer}{2020}]%
        {thread-sanitizer-v2}
ThreadSanitizer \bibinfo{year}{2020}\natexlab{}.
\newblock \bibinfo{title}{ThreadSanitizer}.
\newblock \bibinfo{howpublished}{https://github.com/google/sanitizers}.
  (\bibinfo{year}{2020}).
\newblock


\bibitem[\protect\citeauthoryear{Xie, Xue, and Zhang}{Xie
  et~al\mbox{.}}{2013}]%
        {xie2013acculock}
\bibfield{author}{\bibinfo{person}{Xinwei Xie}, \bibinfo{person}{Jingling Xue},
  {and} \bibinfo{person}{Jie Zhang}.} \bibinfo{year}{2013}\natexlab{}.
\newblock \showarticletitle{Acculock: Accurate and efficient detection of data
  races}.
\newblock \bibinfo{journal}{{\em Software: Practice and Experience\/}}
  \bibinfo{volume}{43}, \bibinfo{number}{5} (\bibinfo{year}{2013}),
  \bibinfo{pages}{543--576}.
\newblock
\showDOI{%
\url{https://doi.org/10.1109/CGO.2011.5764688}}


\bibitem[\protect\citeauthoryear{Yu and Bae}{Yu and Bae}{2016}]%
        {yu2016simplelock+}
\bibfield{author}{\bibinfo{person}{Misun Yu} {and} \bibinfo{person}{Doo-Hwan
  Bae}.} \bibinfo{year}{2016}\natexlab{}.
\newblock \showarticletitle{SimpleLock+: fast and accurate hybrid data race
  detection}.
\newblock \bibinfo{journal}{{\it Comput. J.}} \bibinfo{volume}{59},
  \bibinfo{number}{6} (\bibinfo{year}{2016}), \bibinfo{pages}{793--809}.
\newblock
\showDOI{%
\url{https://doi.org/10.1109/PDCAT.2013.15}}


\bibitem[\protect\citeauthoryear{Yu, Park, Chun, and Bae}{Yu
  et~al\mbox{.}}{2017a}]%
        {yu2017experimental}
\bibfield{author}{\bibinfo{person}{Misun Yu}, \bibinfo{person}{Seung-Min Park},
  \bibinfo{person}{Ingeol Chun}, {and} \bibinfo{person}{Doo-Hwan Bae}.}
  \bibinfo{year}{2017}\natexlab{a}.
\newblock \showarticletitle{Experimental performance comparison of dynamic data
  race detection techniques}.
\newblock \bibinfo{journal}{{\em ETRI Journal\/}} \bibinfo{volume}{39},
  \bibinfo{number}{1} (\bibinfo{year}{2017}), \bibinfo{pages}{124--134}.
\newblock
\showDOI{%
\url{https://doi.org/10.4218/etrij.17.0115.1027}}


\bibitem[\protect\citeauthoryear{Yu, Yang, Su, and Ma}{Yu
  et~al\mbox{.}}{2017b}]%
        {yu2017evaluation}
\bibfield{author}{\bibinfo{person}{Zhen Yu}, \bibinfo{person}{Zhen Yang},
  \bibinfo{person}{Xiaohong Su}, {and} \bibinfo{person}{Peijun Ma}.}
  \bibinfo{year}{2017}\natexlab{b}.
\newblock \showarticletitle{Evaluation and comparison of ten data race
  detection techniques}.
\newblock \bibinfo{journal}{{\em International Journal of High Performance
  Computing and Networking\/}} \bibinfo{volume}{10}, \bibinfo{number}{4-5}
  (\bibinfo{year}{2017}), \bibinfo{pages}{279--288}.
\newblock
\showDOI{%
\url{https://doi.org/10.1504/IJHPCN.2017.086532}}


\end{thebibliography}

\appendix

\section{Predictable Data Races}
\label{sec:predictable-data-races}

We formalize our notion of predictable data races.

\subsection{Run-Time Events and Traces}
\label{sec:run-time-events-and-traces}

We assume concurrent programs making use of shared variables
and acquire/release (a.k.a.~lock/unlock) primitives.
Further constructs such as fork and join are omitted for brevity.
We assume that programs are executed under the sequential consistency
memory model~\citep{Adve:1996:SMC:619013.620590}.
This is a standard assumption made by most data race prediction algorithms.
The upcoming program order condition (see Definition~\ref{def:correct-reordering})
reflects this assumption.

Programs are instrumented to derive a trace of events
when running the program.
A trace is of the following form.

\begin{definition}[Run-Time Traces and Events]
\label{def:run-time-traces-events}
\bda{lcll}
  T & ::= & [] \mid \thread{i}{e} : T   & \mbox{Trace}
  \\ e & ::= &  \readEE{x}{j}
           \mid \writeEE{x}{j}
           \mid \lockEE{y}{j}
           \mid \unlockEE{y}{j}

           & \mbox{Events}
\eda
\end{definition}
Besides $e$, we sometimes use symbols $f$ and $g$ to refer to events.

A trace $T$ is a list of events. We use the notation
a list of objects $[o_1,\dots,o_n]$ is a shorthand
for $o_1:\dots:o_n:[]$. We write $\pp$ to denote the concatenation operator among lists.
For each event $e$, we record the thread id number $i$ in which the event took place,
written $\thread{i}{e}$.
We write $\readEE{x}{j}$ and $\writeEE{x}{j}$
to denote a read and write event on shared variable $x$.
We write $\lockEE{y}{j}$ and $\unlockEE{y}{j}$
to denote a lock and unlock event on mutex $y$.
The number $j$ is distinct for each event
and allows us to uniquely identify each event.
For brevity, we sometimes omit the thread id $i$ and the number $j$.

\begin{example}
\label{ex:tabular-trace-notation}
  We often use a tabular notation for traces where we introduce for each thread
  a separate column and the trace position can be identified via
  the row number.
  Below, we find a trace specified as list of events
  and its corresponding tabular notation.
\bda{c}
  \ba{lcl}
  T & =  & [\thread{1}{\writeEE{x}{1}}, \thread{1}{\lockEE{y}{2}}, \thread{1}{\unlockEE{y}{3}},
     \\ & & \thread{2}{\lockEE{y}{4}}, \thread{2}{\writeEE{x}{5}}, \thread{2}{\unlockEE{y}{6}}]
  \ea
 \\
  \ba{lll}
  & \thread{1}{} & \thread{2}{}
  \\ \hline
  1. & \writeE{x} &
  \\ 2. & \lockE{y} &
  \\ 3. & \unlockE{y} &
  \\ 4. & & \lockE{y}
  \\ 5. & & \writeE{x}
  \\ 6. & & \unlockE{y}
  \ea
  \eda
\end{example}

We define $\compTIDP{T}{e} = j$ if $T=T_1 \pp\ [\thread{j}{e}] \pp\ T_2$ for some traces $T_1, T_2$.
We define $\posP{T}{\thread{i}{e}} = k$ if $\thread{i}{e}$ is the $k$-th event in~$T$.
We often drop the component $T$ and
write $\compTID{e}$ and $\pos{e}$ for short.

We define $\events{T} = \{ e \mid \exists T_1,T_2,j. T = T_1 \pp [\thread{j}{e}] \pp T_2 \}$
to be the set of events in $T$.
We write $e \in T$ if $e \in \events{T}$.

We define $\proj{i}{T} = T'$ the projection of $T$ onto thread $i$
where (1) for each $e \in T$ where $\compTIDP{T}{e} = i$ we have that $e \in T'$, and
(2) for each $e, f \in T'$ where $\posP{T'}{e} < \posP{T'}{f}$ we have that $\posP{T}{e} < \posP{T}{f}$.
That is, the projection onto a thread comprised of all events in that thread
and the program order remains the same.

Traces must be {\em well-formed}: a thread may only acquire an unheld lock and may only release a lock it has acquired.
Hence, for each release event $\thread{i}{\unlockEE{y}{l}}$
there exists an acquire event $\thread{i}{\lockEE{y}{k}}$
where $k < l$ and there is no other acquire on $y$ in between.
We refer to $\thread{i}{\lockEE{y}{k}}$ and $\thread{i}{\unlockEE{y}{l}}$
as a pair of \emph{matching acquire-release} events.
All events $e_1,...,e_n$ in between trace positions $k$ and $l$ must be part of the
thread $i$.

In such a situation, we write
$\cs{i}{\lockEE{y}{k},e_1,\dots,e_n,\unlockEE{y}{l}}$
to denote the events in the \emph{critical section}
represented by the pair
$\thread{i}{\lockEE{y}{k}}$ and $\thread{i}{\unlockEE{y}{l}}$
of matching acquire-release events.

  We write $f \in \cs{i}{\lockEE{y}{k},e_1,\dots,e_n,\unlockEE{y}{l}}$
  if $f$ is one of the events in the critical section.
  We often write $\thread{i}{\csSym{y}}$ as a short-form for a critical section
  $\cs{i}{\lockEE{y}{k},e_1,\dots,e_n,\unlockEE{y}{l}}$.
  We write $\thread{i}{\csSym{y}} \in T$ to denote that the critical section
  is part of the trace~$T$.
  We write $\thread{i}{\acq{\csSym{y}}}$ to refer to $\lockEE{y}{k}$
  and $\thread{i}{\rel{\csSym{y}}}$ to refer to $\unlockEE{y}{l}$.
  If the thread id does not matter, we write $\csSym{y}$ for short and so on.
  If the lock variable does not matter, we write $\csSymAny$ for short and so on.

  We define $\rwTx$ as the set of all read/write
  events in $T$ on some variable $x$.
  We define $\rwTr$ as the union of $\rwTx$ for all variables $x$.


  Let $e, f \in \rwTx$ where $e$ is a read event and $f$ is a write event.
  We say that $f$ is the \emph{last write} for $e$ w.r.t.~$T$ if
  (1) $f$ appears before $e$ in the trace, and
  (2) there is no other write event on $x$ in between $f$ and $e$ in the trace.

\subsection{Trace Reordering}

A trace represents one possible interleaving of concurrent events.
Based on this trace,  we wish to explore alternative interleavings.
In theory, there can be as many interleavings as there are permutations of the original trace.
However, not all permutations are feasible in the sense that they
could be reproduced by executing the program again.

We wish to characterize feasible alternative interleavings
without having to take into account the program.
For this purpose, we assume some idealistic execution scheme:
(1) The program order as found in each thread
is respected. (2) Every read sees the same (last) write.
(3) The lock semantics is respected so that execution will not get stuck.

\begin{definition}[Correct Reordering]
\label{def:correct-reordering}
Let $T$ be a well-formed trace.
Then, trace $T'$ is a \emph{correctly reordered prefix} of $T$ iff
the following conditions hold:
\begin{itemize}
\item Program order:  For each thread id~$i$ we have that $\proj{i}{T'}$ is a subtrace of $\proj{i}{T}$.
\item Last writer: For each read event $e$ in $T'$ where $f$ is the last write for $e$ w.r.t.~$T$,
      we have that $f$ is in $T'$ and $f$ is also the last write for $e$ w.r.t.~$T'$
\item Lock semantics: For $e_1, e_2$ be two acquire events on the same lock
  where $\posP{T'}{e_1} < \posP{T'}{e_2}$
  we have that $\posP{T'}{e_1} < \posP{T'}{f_1} < \posP{T'}{e_2}$
  where $f_1$ is $e_1$'s matching release event.
\end{itemize}
\end{definition}
A correctly reordered trace is a permutation of the original trace
that respects the idealistic execution scheme.
As we will see, a data race may only reveal itself for some prefix.

Critical sections represent atomic units and the events within
cannot be reordered. However, critical sections themselves may be reordered.
We distinguish between schedules that leave the order of critical sections unchanged (trace-specific schedule),
and schedules that reorder critical sections (alternative schedule).

\begin{definition}[Schedule]
  Let $T$ be a well-formed trace
  and $T'$ some correctly reordered prefix of $T$.

  We say $T'$ represents the \emph{trace-specific} \emph{schedule} in $T$
  if the relative position of (common) critical sections (for the same lock variable)
  in $T'$ and $T$ is the same. For lock variable $y$ and critical sections
  $\csSym{y}_1, \csSym{y}_2 \in T$ where $\csSym{y}_1$ appears before $\csSym{y}_2$ in $T$
  we have that $\csSym{y}_1, \csSym{y}_2 \in T'$ and $\csSym{y}_1$ appears before $\csSym{y}_2$ in $T'$.
  Otherwise, we say $T'$ that represents some \emph{alternative schedule}.
\end{definition}

\begin{example}
  Consider the well-formed trace
  \bda{lcl}
   T  & = & [\thread{1}{\writeEE{x}{1}}, \thread{1}{\lockEE{y}{2}}, \thread{1}{\unlockEE{y}{3}},
    \\ & & \thread{2}{\lockEE{y}{4}}, \thread{2}{\writeEE{x}{5}}, \thread{2}{\unlockEE{y}{6}}].
  \eda
  Then,
  \bda{lcl}
   T'  & = & [\thread{2}{\lockEE{y}{4}}, \thread{2}{\writeEE{x}{5}}, \thread{1}{\writeEE{x}{1}},
        \\ & & \thread{2}{\unlockEE{y}{6}}, \thread{1}{\lockEE{y}{2}}, \thread{1}{\unlockEE{y}{3}}]
   \eda
  is a correctly reordered prefix of $T$ where
  $T'$ represents an alternative schedule.
\end{example}

\subsection{Data Race}

A data race is represented as a pair $(e, f)$ of events where $e$ and $f$ are in conflict
and we find a correctly reordered prefix (schedule)
where $e$ appears right before $f$ in the trace.

The condition that $e$ appears right before $f$ is useful to clearly
distinguish between write-read and read-write races.
We generally assume that for each read there is an initial write.
Write-read race pairs are linked to write-read dependencies where a write immediately precedes
a read. Read-write race pairs indicate situations where a read might interfere
with some other write, not the read's last write.
For write-write race pairs $(e,f)$ it turns out that if $e$ appears right before $f$
for some reordered trace
then $f$ can also appear right before $e$ by using a slightly different reordering.
Hence, write-write pairs $(e,f)$ and $(f,e)$ are equivalent
and we only report the representative $(e,f)$ where $e$ appears before $f$ in the original trace.

\begin{definition}[Initial Writes]
  We say a trace $T$ satisfies the \emph{initial write} property
  if for each read event $e$ on variable $x$ in $T$ there exists
  a write event $f$ on variable $x$ in $T$ where $\posP{T}{f} < \posP{T}{e}$.
\end{definition}

The initial write of a read does not necessarily need to occur within the same thread.
It is sufficient that the write occurs before the read in the trace.
From now on we assume that all traces satisfy the initial write assumption,
as well as the well-formed property.

\begin{definition}[Predictable Data Race Pairs]
\label{def:pred-data-race}
  Let $T$ be a trace.
  Let $T'$ be a correctly reordered prefix of $T'$.
  Let $e, f \in T$.
  We refer to $(e, f)$ as a \emph{predictable data race pair}
  if (a) $e, f$ are two conflicting events in $T$, and
  (b) $e$ appears right before $f$ in the trace $T'$.
  We refer to $T'$ as \emph{witness}.

  We say $(e, f)$ is a \emph{write-read} race pair if $e$ is a write and $f$ is a read.
  We say $(e, f)$ is a \emph{read-write} race pair if $e$ is a read and $f$ is a write.
  We say $(e, f)$ is a \emph{write-write} race pair if both events are writes.

  We write $\dataRace{T}{T'}{e}{f}$ for predictable write-read, read-write and write-write race pairs
  and traces $T$ and $T'$ as specified above.
  For write-write pairs $(e, f)$ we demand that $\posP{T}{e} < \posP{T}{f}$.

  We define $\allPredRacesP{T} = \{ (e,f) \mid e,f \in T \wedge \exists T'. \prefixOf{T}{T'} \wedge \dataRace{T}{T'}{e}{f} \}$.
  We refer to $\allPredRacesP{T}$ as the set of \emph{all predictable data pairs}
  derivable from $T$.

  We define  $\schedSpecificPredRacesP{T} = \{ (e,f) \mid e,f \in T \wedge \exists T'. \prefixOf{T}{T'} \wedge \dataRace{T}{T'}{e}{f}
  \wedge \mbox{$T'$ trace-specific schedule} \}$.
  We refer to $\schedSpecificPredRacesP{T}$ as the set of \emph{all trace-specific predictable data race pairs}
  derivable from $T$.
\end{definition}
Our definition of predictable races
follows~\citep{10.1145/3360605,Mathur:2018:HFR:3288538:3276515}.
and is more general compared to earlier definitions as found in~\citep{Smaragdakis:2012:SPR:2103621.2103702,Kini:2017:DRP:3140587.3062374}. The difference is that~\citep{Smaragdakis:2012:SPR:2103621.2103702,Kini:2017:DRP:3140587.3062374}
only consider the 'first' race as a predictable race
whereas~\citep{10.1145/3360605,Mathur:2018:HFR:3288538:3276515} also consider
'subsequent' races as predictable races.
Identifying races beyond the first race is useful as we explain
via the following example.


\begin{example}
  \label{ex:pred-race-pairs}
  Consider the following trace $T$ where we use the tabular notation.
  \bda{llll}
  & \thread{1}{} & \thread{2}{} & \thread{3}{}
  \\ \hline
  1. & \writeE{x} &&
  \\ 2. && \writeE{x} &
  \\ 3. && \readE{x} &
  \\ 4. &&& \readE{x}
  \\ 5. &&& \writeE{x}
  \eda
  For each event $e$ we consider the possible candidates $f$
  for which $(e, f)$ forms a predictable race pair.
  We start with event $\writeEE{x}{1}$.

  For $\writeEE{x}{1}$ we immediately find (1) $(\writeEE{x}{1}, \writeEE{x}{2})$.
  We also find (2) $(\writeEE{x}{1}, \writeEE{x}{5})$ by putting
  $\writeEE{x}{1}$ in between $\readEE{x}{4}$ and $\writeEE{x}{5}$.
  There are no further combinations $(\writeEE{x}{1}, f)$
  where $\writeEE{x}{1}$ can appear right before some $f$.
  For instance, $(\writeEE{x}{1}, \readEE{x}{3})$ is not valid because
  otherwise the last writer condition
  in Definition~\ref{def:correct-reordering} is violated.

  Consider $\writeEE{x}{2}$.
  We find (3) $(\writeEE{x}{2}, \writeEE{x}{1})$ because
  $$T' = [\writeEE{x}{2}, \writeEE{x}{1}]$$
  is a correctly reordered prefix of $T$.
  It is crucial that we only consider prefixes. Any extension of $T'$
  that involves $\readEE{x}{3}$ would violate the last writer condition
  in Definition~\ref{def:correct-reordering}.
  For $\writeEE{x}{2}$ there is another pair (4) $(\writeEE{x}{2}, \readEE{x}{4})$.
  The pair $(\writeEE{x}{2}, \readEE{x}{3})$ is not a valid write-read race pair
  because $\writeEE{x}{2}$ and $\readEE{x}{3}$ result from the same thread
  and therefore are not in conflict.

  Consider $\readEE{x}{3}$.
  We find pairs (5) $(\readEE{x}{3}, \writeEE{x}{1})$ and (6) $(\readEE{x}{3}, \writeEE{x}{5})$.
  For instance (5) is due to the prefix
  $$[\writeEE{x}{2}, \readEE{x}{3}, \writeEE{x}{1}].$$
  The remaining race pairs are
  (7) $(\readEE{x}{4}, \writeEE{x}{1})$ and (8) $(\writeEE{x}{5}, \writeEE{x}{1})$.

  Pairs (1) and (3) as well as pairs (2) and (8) are equivalent write-write race pairs.
  When collecting all predictable race pairs we only keep the representatives (1) and (2).
  Hence, we find
  $\allPredRacesP{T} = \{ (1), (2), (4), (5), (6), (7) \}$
  where each race pair is represented by the numbering schemed introduced above.
  There are no critical sections and therefore no alternative schedules.
  Hence, $\allPredRacesP{T} = \schedSpecificPredRacesP{T}$.
\end{example}

\citet{Smaragdakis:2012:SPR:2103621.2103702,Kini:2017:DRP:3140587.3062374} identify race pair (1) as the first race pair.
All race pairs (1-7) are schedulable races
according to~\citet{Mathur:2018:HFR:3288538:3276515}.
For example, consider (4) $(\writeEE{x}{2}, \readEE{x}{4})$.
and (6) $(\readEE{x}{3}, \writeEE{x}{5})$.
A witness for (6) is $T' = [\writeEE{x}{2}, \readEE{x}{4}, \readEE{x}{3}, \writeEE{x}{5}]$.
In $T'$ there is the 'earlier' race (4) and there is no other witness for (6)
that does not contain (4).
So it seems that (6) is not a 'real' race because
after (4) the program's behavior may become undefined.

However, it is easy to fix earlier races.
We make the conflicting
events mutually exclusive by introducing a fresh lock variable.
In terms of the original trace, we replace subtrace
$[\thread{2}{\writeEE{x}{2}}]$
by $$[\thread{1}{\lockE{y}}, \thread{2}{\writeEE{x}{2}}, \thread{1}{\unlockE{y}}]$$
and subtrace $[\thread{3}{\readEE{x}{4}}]$
by $$[\thread{2}{\lockE{y}}, \thread{3}{\writeEE{x}{4}}, \thread{2}{\unlockE{y}}]$$
where $y$ is a fresh lock variable.
Race (6) becomes then a real race.
Hence, the motivation to consider all races
as we otherwise require multiple execute-report-fix cycles.

The next example highlights the fact that a race may only reveal
itself for some prefix.
\begin{example}
  Consider
  \bda{lll}
  & \thread{1}{} & \thread{2}{}
  \\ \hline
  1. & \writeE{y} &
  \\ 2. & \lockE{z} &
  \\ 3. & \writeE{x} &
  \\ 4. & \unlockE{z} &
  \\ 5. && \lockE{z}
  \\ 6. && \writeE{y}
  \\ 7. && \readE{x}
  \\ 8. && \unlockE{z}
  \eda

  There is one predictable race $(\writeEE{y}{1}, \writeEE{y}{6})$
  that results from some alternative schedule.
  Consider
  $$T' = [\thread{2}{\lockEE{z}{5}}, \thread{1}{\writeEE{y}{1}},\thread{2}{\writeEE{y}{6}}].$$
  There is no extension of $T'$ that covers all events in $T$ as otherwise
  we would violate the last writer condition.
\end{example}

We summarize.
For each race pair $(e,f)$ there is a reordering
where $e$ appears right before $f$ in the reordered trace.
Each write-write race pair $(e,f)$ is also a write-write race pair $(f,e)$.
We choose the representative $(e,f)$ where $e$ appears before $f$ in the original trace.
For each write-read race pair $(e,f)$ we have that $e$ is $f$'s last write.
Each read-write race pair $(e,f)$ represents a situation
where the read $e$ can interfere with some other write $f$.
Formal statements see below.


\begin{lemma}
Let $T$ be some trace and
$(e, f)$ be some write-write race pair for $T$.
Then, we have that $(f, e)$ is also a write-write race pair for $T$.
\end{lemma}
\begin{proof}
  By assumption $T'$ is some correctly reordered prefix
  where $T' = [\dots,e,f]$.
  We can reorder $e$ and $f$ in $T'$ while maintaining
  the conditions in Definition~\ref{def:correct-reordering}.
  Thus, we are done.
\end{proof}

\begin{lemma}
Let $T$ be some trace and
$(e, f)$ be some write-read race pair for $T$.
Then, $(f, e)$ cannot be a read-write race pair for $T$.
\end{lemma}
\begin{proof}
  By construction $e$ must be $f$'s `last write'.
  Hence, $(f, e)$ is not valid as otherwise the `last write' property is violated.
\end{proof}

\begin{lemma}
Let $T$ be some trace and
$(e, f)$ be some read-write race pair for $T$.
Then, $(f, e)$ cannot be a write-read race pair for $T$.
\end{lemma}
\begin{proof}
  For this result we rely on the initial writes assumption.
  For the read-write race pair $(e, f)$ we know that
  $f$ is \emph{not} $e's$ `last write'.
  Then, $(f, e)$ is not valid. If it would then $f$ is  $e's$ `last write'. Contradiction.
\end{proof}

From above we conclude that for each write-read race pair $(e,f)$
we have that $e$ appears before $f$ in the original trace $T$.
For read-write race pairs $(e,f)$, $e$ can appear before or after $f$
in the original trace.
See cases (5) and (6) in Example~\ref{ex:pred-race-pairs}.

\section{Fork and Join}

\begin{algorithm}
\caption{\PWREE{} algorithm (first pass) with fork and join} \label{alg:fork-join}

\begin{algorithmic}[1]
  \Function{w3}{$V, \LStSym$}
    \For {$y \in \LStSym$}
    \For {$(\thread{j}{k}, V') \in \LH{y}$}
     \If {$k < \accVC{V}{j}$}
     \State $V = V \sqcup V'$
     \EndIf
     \EndFor
     \EndFor

    \Return V
    \EndFunction
\end{algorithmic}

\begin{algorithmic}[1]
  \Procedure{acquire}{$i,y$}
  \State $\threadVC{i} = \Call{w3}{\threadVC{i},\LSt{i}}$
\State $\LSt{i} = \LSt{i} \cup \{y\}$
\State $\Acq{y} = \thread{i}{\accVC{\threadVC{i}}{i}}$
\State $\incC{\threadVC{i}}{i}$
\EndProcedure
\end{algorithmic}

\begin{algorithmic}[1]
  \Procedure{release}{$i,y$}
  \State $\threadVC{i} = \Call{w3}{\threadVC{i},\LSt{i}}$
\State $\LSt{i} = \LSt{i} - \{x\}$
\State $\LH{y} = \LH{y} \cup \{(\Acq{y}, \threadVC{i})\}$
\State $\incC{\threadVC{i}}{i}$
\EndProcedure
\end{algorithmic}

\begin{algorithmic}[1]
  \Procedure{write}{$i,x$}
  \State $\threadVC{i} = \Call{w3}{\threadVC{i},\LSt{i}}$
  \State $\vcEvt = \{ (\thread{i}{\accVC{\threadVC{i}}{i}}, \threadVC{i}, \LSt{i}) \} \cup \vcEvt$
\State $\edges{x} =
       \{ \thread{j}{k} \gtEdge \thread{i}{\accVC{\threadVC{i}}{i}}
       \mid \thread{j}{k} \in \rwVC{x} \wedge
       k < \accVC{\threadVC{i}}{j} \} \cup \edges{x}$
\State $\concEvt{x} = \{ (\thread{j}{k}, \thread{i}{\accVC{\threadVC{i}}{i}})
    \mid \thread{j}{k} \in \rwVC{x}
          \wedge k > \accVC{\threadVC{i}}{j} \} \cup \concEvt{x}$
  \State $\rwVC{x} = \{ \thread{i}{\accVC{\threadVC{i}}{i}} \}
         \cup \{ \thread{j}{k} \mid \thread{j}{k} \in \rwVC{x} \wedge
         k > \accVC{\threadVC{i}}{j} \}$
\State $\lastWriteVC{x} = \threadVC{i}$
\State $\lastWriteVCt{x} = i$
\State $\lastWriteVCL{x} = \LSt{i}$
\State $\incC{\threadVC{i}}{i}$
\EndProcedure
\end{algorithmic}

\begin{algorithmic}[1]
  \Procedure{fork}{$i,j$}
\State $\threadVC{j} = \threadVC{i}$
\State $\incC{\threadVC{i}}{i}$
\EndProcedure
\end{algorithmic}

\begin{algorithmic}[1]
  \Procedure{join}{$i,j$}
\State $\threadVC{i} = \threadVC{j} \sqcup \threadVC{i}$
\State $\incC{\threadVC{i}}{i}$
\EndProcedure
\end{algorithmic}

\end{algorithm}

Algorithm \ref{alg:fork-join} extends Algorithm \ref{alg:w3poee-firstpass} with fork and join.

\section{Additional Examples}

\begin{example}
  \label{ex:wwwpoee-simple}
  We consider a run of the first pass of $\PWREE$ for the following trace.
  Instead of epochs, we write $w_i$ for a write at trace position $i$.
  A similar notation is used for reads.
  We annotate the trace with $\rwVC{x}$, $\edges{x}$ and $\concEvt{x}$.
  For $\edges{x}$ and $\concEvt{x}$ we only show incremental updates.
  For brevity, we omit the set $\vcEvt$ because locksets and vector clocks
  of events do not matter here.

  \bda{llllll}
  & \thread{1} & \thread{2} & \rwVC{x} & \edges{x} & \concEvt{x}
  \\ \hline
  1. & \writeE{x} && \{ w_1 \} &&
  \\ 2. & \writeE{x} && \{ w_2 \} & w_1 \gtEdge w_2 &
  \\ 3. & & \writeE{x} & \{ w_2, w_3 \} && (w_2,w_3)
  \\ 4. & & \readE{x}  & \{ w_2, r_4 \} & w_3 \gtEdge r_4 & (w_2, r_4)
  \eda
  The potential races covered by $\concEvt{x}$ are $(w_2, w_3)$ and $(r_4, w_2)$.
  These are also predictable races.
  As said, the set $\concEvt{x}$ follows the trace position order.
  Hence, we find $(w_2, r_4) \in \concEvt{x}$.
  Overall, there are four predictable races.
  The first pass of $\PWREE$, i.e.~the set $\concEvt{x}$, fails to capture the predictable races
  $(w_1, w_3)$ and $(r_4, w_1)$.

  The missing pairs can be obtained via the second pass as follows.
  Starting from $(w_2,w_3) \in\concEvt{x}$ via
  $w_1 \gtEdge w_2 \in\edges{x}$ we can reach $(w_1, w_3)$.
  From $(w_2, r_4) \in\concEvt{x}$ via
  $w_1 \gtEdge w_2 \in \edges{x}$ we reach $(w_1, r_4)$.
  The pair $(w_1, r_4)$ represents a read-write pair.
  When reporting this pair we simply switch the order of events.
\end{example}

\section{Proofs of Results in Main Text}

\subsection{Auxiliary Results}

\begin{lemma}
$\shbSym \not\subseteq \wcpSym$.
\end{lemma}
\begin{proof}
Consider Example~\ref{ex:hb-unsound}.
\end{proof}

\begin{lemma}
 \label{le:wcp-subset-shb}
$\shbSym \subseteq \wcpSym$.
\end{lemma}
\begin{proof}
  Both relations apply the PO condition.

  Consider the `extra' WCP conditions.
  These conditions relax the RAD condition.
  Hence, if any of these WCP conditions apply,
  the RAD condition applies as well.
\end{proof}

\begin{lemma}
\label{le:wrd-cs-order}
Let $T$ be a  trace.
Let $<$ denote some strict partial order among elements in $T$.
Let $e, f \in T$, $\csSym{y}_1$ and $\csSym{y}_2$ be two critical sections for the same lock variable $y$
such that (1) $\acq{\csSym{y}_1} < e < \rel{\csSym{y}_1}$,
          (2) $\acq{\csSym{y}_2} < f < \rel{\csSym{y}_2}$, and
(3) $e < f$.
Then, we have that $\neg (\rel{\csSym{y}_2} < \acq{\csSym{y}_1})$.
\end{lemma}
\begin{proof}
  Suppose, $\rel{\csSym{y}_2} < \acq{\csSym{y}_1}$.
  Then, we find that $\acq{\csSym{y}_1} < e < f < \rel{\csSym{y}_2} < \acq{\csSym{y}_1}$.
  This is a contradiction and we are done.
\end{proof}

\subsection{Proof of Proposition~\ref{prop:pwr-vs-wdp}}

\begin{proof}
We first show that $\wdpSym \subseteq \pwrSym$.
\PWR\ applies PO like WDP.
The WDP rule RCD is an instance of the \PWR\ rule ROD
in combination with the WRD rule.
Similarly, the WDP rule RRD is an instance of the
\PWR\ rule ROD
in combination with the PO rule.

Example~\ref{ex:pwr-is-stronger-than-wdp}
shows that the reverse direction
$\pwrSym \subseteq \wdpSym$ does not hold.
\end{proof}

\subsection{Proof of Proposition~\ref{prop:lockset-www-completeness}}

\begin{proof}
  We need to show that the $\pwrSym$ relation does not rule out any predictable data race pairs.
  For this to hold we show that any correctly reordered prefix satisfies the $\pwrSym$ relation.
  Clearly, this is the case for the PO and WRD.

  What other happens-before conditions need to hold for correctly reordered prefixes?
  For critical sections we demand that they must follow a proper acquire/release order.
  We also cannot arbitrarily reorder critical sections as write-read dependencies
  must be respected.  See Lemma~\ref{le:wrd-cs-order}.
  Condition ROD catches such cases.

  We have $e \in \csSym{y}$, $f \in \csSym{y}'$ and $\pwr{e}{f}$.
  Critical section $\csSym{y}'$ appears after $\csSym{y}$ (otherwise
  $\pwr{e}{f}$ would not hold).
  Considering the entire trace, $\csSym{y}'$ cannot be put in front of $\csSym{y}$
  via some reordering (see Lemma~\ref{le:wrd-cs-order}).

  As we may only consider a prefix, it is legitimate to apply some reordering
  that only affects parts of $\csSym{y}'$. Due to $\pwr{e}{f}$
  we may only reorder the part of $\csSym{y}'$ that is above of $f$ in the trace.
  This requirement is captured via $\pwr{\rel{\csSym{y}}}{f}$.

  We find that the $\pwrSym$ relation does not rule out any of the
  correctly reordered prefixes. This concludes the proof.
\end{proof}

\subsection{Proof of Proposition~\ref{prop:lockset-www-two-threads-soundness}}

\begin{proof}
  We need to show that some correctly reordered prefix of $T$ exists
  for which the potential Lockset-\PWR\ race pair $(e,f)$
  appear right next to each other in the reordered trace.
  W.l.o.g.~we assume that $e$ appears before $f$ in $T$
  and $\compTID{e}=1$ and $\compTID{f}=2$.

  By assumption
  $\LS{e} = \LS{f} = \{ \}$.
  The layout of the trace is as follows.

  \bda{l|l}
  \thread{1}{}
  & \thread{2}{}
  \\ \hline
  \vdots & \vdots
  \\ e &
  \\ T_1 &
  \\ & T_1'
  \\ T_2 &
  \\ & T_2'
  \\ \vdots & \vdots
  \\ T_n &
  \\ & T_n'
  \\ & f
  \eda
  Clearly, none of the parts $T_1, \dots, T_n$ can happen before any of the
  parts $T_1', \dots, T_n'$ w.r.t.~the $\pwrSym$ relation.
  Otherwise, $e \pwrSym f$ which contradicts the assumption.

  Hence, $T_1', \dots, T_n'$ are independent of $T_1, \dots, T_n$
  and the trace can be correctly reordered as follows.
\bda{l|l}
  \thread{1}{}
  & \thread{2}{}
  \\ \hline
  \vdots & \vdots
  \\ & T_1'
  \\ & \vdots
  \\ & T_n'
  \\ e &
  \\ & f
  \\ T_1 &
  \\ \vdots
  \\ T_n &
  \eda
  Hence, we are done.
\end{proof}

The result does not extend to more than two threads.
The condition that the lockset is empty is also critical.

\begin{example}
\label{ex:flat-lock-but-unsound-under-pwr}
Consider
\bda{lll}
& \thread{1} & \thread{2}
\\ \hline
1.  &&      \writeE{x}
\\ 2. & \lockE{z} &
\\ 3. & \readE{x} &
\\ 4. & \writeE{y} &
\\ 5. & \unlockE{z} &
\\ 6.     &&   \lockE{z}
\\ 7.     &&   \writeE{x}
\\ 8.     &&   \unlockE{z}
\\ 9.     &&   \writeE{y}
\eda
Events $\writeEE{y}{4}$ and $\writeEE{y}{9}$ are not ordered
under \PWR. The lockset of $\writeEE{y}{4}$ contains $z$.
Both events are a potential lockset-\PWR\ race pair
but this is not a predictable data race pair.
\end{example}

\subsection{Proof of Proposition~\ref{prop:w3po-post}}

We first state some auxiliary results.

In general, we can reach all missing pairs by using pairs in $\concEvt{x}$
as a start and by following edge constraints.
This property is guaranteed by the following statement.
We slightly abuse notation and identify events $e, f, g$ via their epochs
and vice versa.

\begin{lemma}
\label{le:reach-all-concurrent-w3-pairs}
  Let $T$ be a  trace and $x$ be some variable.
  Let $\edges{x}$ and $\concEvt{x}$ be obtained by \PWREE.
  Let $(e, f)$ be two conflicting events involving variable $x$
  where $(e, f) \not\in \concEvt{x}$, $\pos{e} < \pos{f}$ and
  $e, f$ are concurrent to each other w.r.t.~\PWR.
  Then, there exists $g_1,\dots g_n \in \edges{x}$ such that
  $e \gtEdge g_1 \gtEdge \dots \gtEdge g_n$
  and $(g_n,f) \in \concEvt{x}$.
\end{lemma}
\begin{proof}
  We consider the point in time event $e$ is added to $\rwVC{x}$
  when running \PWREE.
  By the time we reach $f$, event $e$ has been removed from $\rwVC{x}$.
  Otherwise, $(e, f) \in \concEvt{x}$ which
  contradicts the assumption.

  Hence, there must be some $g_1$ in $\rwVC{x}$
  where
  $$
  \pos{e} < \pos{g_1} < \pos{f}.
  $$
  As $g_1$ has removed $e$, there must exist $e \gtEdge g_1 \in \edges{x}$ (1).

  By the time we reach $f$, either $g_1$ is still in $\rwVC{x}$,
  or $g_1$ has been removed by some $g_2$
  where $g_1 \gtEdge g_2 \in \edges{x}$ and $g_2 \in \rwVC{x}$.
  As between $e$ and $f$ there can only be a finite number of events,
  we must reach some $g_n \in \rwVC{x}$
  where $g_1 \gtEdge \dots \gtEdge g_n$ (2).
  Event $g_n$ must be concurrent to $f$.

  Suppose $g_n$ is not concurrent to $f$. Then, $g_n \pwrSym f$ (3).
  The case $f \pwrSym g_n$ does not apply because $g_n$ appears before $f$ in the trace.
  Edges imply \PWR\ relations. From (2), we conclude that
  $g_1 \pwrSym \dots \pwrSym g_n$ (4).
  (1), (2) and (4) combined yields $e \pwrSym f$. This contradicts the assumption
  that $e$ and $f$ are concurrent.

  Hence, $g_n$ is concurrent to $f$.
  Hence, $(g_n, f) \in \concEvt{x}$.
  Furthermore, we have that $e \gtEdge g_1 \gtEdge \dots \gtEdge g_n \in \edges{x}$.
\end{proof}

Example~\ref{ex:w3po-post} does not contradict the above
Lemma~\ref{le:reach-all-concurrent-w3-pairs}.
The lemma states that all concurrent pairs can be identified.

The next property characterizes a sufficient condition under which a pair
is added to $\concEvt{x}$.

\begin{lemma}
\label{le:no-other-conc-in-between}
Let $T$ be a well-formed trace.
Let $e,f \in \rwTx$ for some variable $x$
such that (1) $e$ and $f$ are concurrent to each other w.r.t.~\PWR,
(2) $\pos{f} > \pos{e}$, and
(3) $\neg\exists g \in \rwTx$ where $g$ and $f$ are concurrent
to each other w.r.t.~\PWR\ and $\pos{f} > \pos{g} > \pos{e}$.
Let $\concEvt{x}$ be the set obtained by $\PWREE$.
Then, we find that $(e,f) \in \concEvt{x}$.
\end{lemma}
\begin{proof}
  By induction on $T$. Consider the point where $e$ is added to $\rwVC{x}$.
  We assume that $e$'s epoch is of the form $\thread{j}{k}$.
  We show that $e$ is still in $\rwVC{x}$
  at the point in time we process $f$.

  Assume the contrary. So, $e$ has been removed from $\rwVC{x}$.
  This implies that there is some $g$ such that
  $e \pwrSym g$ and $\pos{f} > \pos{g} > \pos{e}$.
  We show that $g$ must be concurrent to $f$.

  Assume the contrary.  Suppose $g \pwrSym f$. But then $e \pwrSym f$
  which contradicts the assumption that $e$ and $f$ are concurrent
  to each other.
  Suppose $f \pwrSym g$. This contradicts the fact that $\pos{f} > \pos{g}$.

  We conclude that $g$ must be concurrent to $f$.
  This is a contradiction to (3).
  Hence, $e$ has not been removed from $\rwVC{x}$.

  By assumption $e$ and $f$ are concurrent to each other.
  Then, we can argue that $k > \accVC{\threadVC{i}}{j}$
  where by assumption $\threadVC{i}$ is $f$'s vector clock
  and $e$ has the epoch $\thread{j}{k}$.
  Hence, $(e,f)$ is added to $\concEvt{x}$.
\end{proof}

We are now ready to verify Proposition~\ref{prop:w3po-post}.

\begin{proof}
  We first show that the construction of $\accConcs{x}$ terminates by
  showing that no pair is added twice.
  Consider $(e,f) \in \concEvt{x}$ where $g \gtEdge e$.
  We remove $(e,f)$ and add $(g,f)$.

  Do we ever encounter $(f,e)$? This is impossible as the position of first component
  is always smaller than the position of the second component.

  Do we re-encounter $(e,f)$? This implies that there must exist $g$
  such that $e \gtEdge g$ where $(g,f) \in \concEvt{x}$.
  By Lemma~\ref{le:no-other-conc-in-between} this is in contradiction
  to the assumption that $(e,f)$ appeared in $\concEvt{x}$.
  We conclude that the construction of $\accConcs{x}$ terminates.

  Pairs are kept in a total order imposed by the position of the first component.
  As shown above we never revisit pairs.
  For each $e$ any predecessor $g$ where $g \gtEdge e \in \edges{x}$
  can be found in constant time (by using a graph-based data structure).
  Then, a new pair is built in constant time.

  There are $O(n*n)$ pairs overall to consider.
  We conclude that the construction of $\accConcs{x}$ takes time $O(n*n)$.
  By Lemma~\ref{le:reach-all-concurrent-w3-pairs} we can guarantee that
  all pairs in $\allConcsP{T}{x}$ will be reached.
  Then, $\allConcsP{T}{x} \subseteq \accConcs{x}$.
\end{proof}

\subsection{Proof of Lemma~\ref{le:lockset-w3-filtering}}

\begin{proof}
  Follows from the fact that $\PWREE$ computes the event's lockset and vector clock.
  To check if two events are concurrent it suffices to compare the earlier in the trace
  events time stamp against the time stamp of the later in the trace event.
  Recall that for pairs in $\concEvt{x}$ and therefore also $\accConcs{x}$,
  the left component event occurs earlier in the trace than the right component event.
\end{proof}

\section{WRD Race Pairs}
\label{sec:wrd-race-pairs}

Lockset-\PWR\ WRD race pairs characterize write-read races resulting from the trace-specific
or alternative schedules.

\begin{example}
 \label{ex:w-w-and-w-r-race}
Consider the following trace (on the left) and the set of predictable and trace-specific race pairs (on the right).
  \bda{lll}
  & \thread{1}{} & \thread{2}{}
  \\ \hline
  1. & & \writeE{x}
  \\ 2. & \writeE{x} &
  \\ 3. & \lockE{y} &
  \\ 4. & \unlockE{y} &
  \\ 5. && \lockE{y}
  \\ 6. && \unlockE{y}
  \\ 7. && \readE{x}
  \eda
  where
  \bda{lcl}
  \allPredRacesP{T} & = & \{ (\writeEE{x}{1}, \writeEE{x}{2}), (\writeEE{x}{2}, \readEE{x}{7}) \}
  \\ \\ \schedSpecificPredRacesP{T} & = & \{ (\writeEE{x}{1}, \writeEE{x}{2}) \}
  \eda
  There are no read-write races in this case.
  The pair $(\writeEE{x}{2}, \readEE{x}{7})$ results from the correctly
  reordered prefix (alternative schedule)
   $T' =
      [ \thread{2}{\writeEE{x}{1}},
        \thread{2}{\lockEE{y}{5}},
        \thread{2}{\unlockEE{y}{6}},
        \thread{1}{\writeEE{x}{2}},
          \thread{2}{\readEE{x}{7}} ].
   $
      The pair $(\writeEE{x}{2}, \readEE{x}{7})$ is not in $\schedSpecificPredRacesP{T}$
      because $T'$ represents some alternative schedule and there is no trace-specific
      schedule where the write and read appear right next to each other.
\end{example}

The pair $(\writeEE{x}{1}, \readEE{x}{7})$ is Lockset-\PWR\ WRD race pair.
However, this pair is not a SHB WRD race pair because the write-read race
results from some alternative schedule.

\section{\PWR\ Variants}
\label{sec:w3-variants}

We consider the following variant of \PWR\ where
we impose a slightly different ROD rule.

\begin{definition}[WRD + ROD with Acquire]
\label{def:wwwa-relation}
  Let $T$ be a trace.
  We define a relation $\pwra{}{}$ among trace events
  as the smallest partial order that
  satisfies conditions PO and WRD as well as the following condition:

  \begin{description}
  \item[ROD with Acquire:] Let $f \in T$ be an event.
    Let $\csSym{y}$, $\csSym{y}'$ be two critical sections
    where $\csSym{y}$ appears before $\csSym{y}'$ in the trace, $f \in \csSym{y}'$
    and $\pwra{\acq{\csSym{y}}}{f}$.
    Then, $\pwr{\rel{\csSym{y}}}{f}$.
  \end{description}

  We refer to $\pwra{}{}$ as the \emph{WRD + ROD with Acquire} (PWRA)
  relation.
\end{definition}

The ROD rule  in Definition~\ref{def:www-relation} is more general
compared to the ROD with Acquire rule.
The ROD rule says that if $e \in \csSym{y}$, $f \in \csSym{y}'$ and $\pwr{e}{f}$.
then $\pwr{\rel{\csSym{y}}}{f}$. Hence, the ROD with Acquire rule is an instance
of this rule. Take $e = \acq{\csSym{y}}$.
Hence, $\pwraSym \subseteq \pwrSym$.
We can even show that all \PWR\ relations are already covered by PWRA.

\begin{lemma}
$\pwrSym = \pwraSym$.
\end{lemma}
\begin{proof}
  Case $\pwraSym \subseteq \pwrSym$: Follows from the fact that PWRA is an instance of \PWR.

  Case $\pwrSym \subseteq \pwraSym$:
  We verify this case by induction over the number of ROD rule applications.

  The base cases of the induction proof hold as both \PWR\ and PWRA assume PO and WRD.
  Consider the induction step. We must find the following situation.
  We have that $\pwr{\rel{\csSym{y}}}{f}$ where
  (1) $e \in \csSym{y}$, (2) $f \in \csSym{y}'$ and (3) $\pwr{e}{f}$.
  We need to show that $\pwra{\rel{\csSym{y}}}{f}$.

  From (1), (3) and PO we conclude that $\pwr{\acq{\csSym{y}}}{f}$.
  By induction we find that $\pwra{\acq{\csSym{y}}}{f}$.
  We are in the position to apply the ROD with Acquire rule
  and conclude that $\pwra{\rel{\csSym{y}}}{f}$ and we are done.
\end{proof}

We consider yet another variant of \PWR.

\begin{definition}[WRD + ROD for Read]
\label{def:wwwr-relation}
  Let $T$ be a trace.
  We define a relation $\pwrr{}{}$ among trace events
  as the smallest partial order that
  satisfies conditions PO and WRD as well as the following condition:

  \begin{description}
  \item[ROD for Read:] Let $e,f \in T$ be two events
    where $f$ is a read event.
    Let $\csSym{y}$, $\csSym{y}'$ be two critical sections
    where $e \in \csSym{y}$, $f \in \csSym{y}'$
     and $\pwrr{e}{f}$.
    Then, $\pwrr{\rel{\csSym{y}}}{f}$.
  \end{description}

  We refer to $\pwrr{}{}$ as the \emph{WRD + ROD for Read} (PWRR)
  relation.
\end{definition}

The difference to \PWR\ is that the ROD for Read rule only applies to read events.
Again, we find that $\pwrrSym \subseteq \pwrSym$ because
PWRR is an instance of \PWR.
However, the other direction does not hold because
some \PWR\ relations do not apply for PWRR as the following example shows.

\begin{example}
\label{ex:w3r-is-weaker}
  Consider the trace
  \bda{lll}
    & \thread{1}{} & \thread{2}{}
  \\ \hline
  1. & \lockE{y} &
  \\ 2. & \writeE{x} &
  \\ 3. & \writeE{z} &
  \\ 4. & \unlockE{y} &
  \\ 5. && \readE{x}
  \\ 6. && \lockE{y}
  \\ 7. && \writeE{z}
  \\ 8. && \unlockE{y}
  \\ 9. && \writeE{z}
  \eda
  Between $\writeEE{x}{2}$ and $\readEE{x}{5}$ there is a WRD.
  In combination with PO, we find that $\lockEE{y}{1} \pwrSym \writeEE{z}{7}$.
  Via the ROD rule we conclude that $\unlockEE{y}{4} \pwrSym \writeEE{z}{7}$.
  As there is no read event in the (second) critical section
  $(\lockEE{y}{6}, \unlockEE{y}{8})$,
  we do not impose $\unlockEE{y}{4} \pwrSym \writeEE{z}{7}$ under PWRR.
\end{example}

We summarize.
\PWR\ and PWRA are equivalent.
PWRR is weaker. In the context of data race prediction this means
that by using PWRR we may encounter more false positives.

Consider again Example~\ref{ex:w3r-is-weaker}.
Under PWRR, conflicting events $\writeEE{z}{3}$ and $\writeEE{z}{9}$
are not synchronized and their lockset is disjoint.
Hence, $(\writeEE{z}{3}, \writeEE{z}{9})$ form a potential data race pair under PWRR.
This is a false positive because due to the WRD the critical sections cannot be reordered
such that $\writeEE{z}{3}$ and  $\writeEE{z}{9})$ appear right next to each other.

\section{\PWREE{} Optimizations}
\label{sec:wwwpoee-optimizations}

\subsection{Application of ROD Rule}

Function \Call{w3}{} enforces the ROD rule.
In general, this needs to be done for each event to be processed.
For events in thread~$i$,
we can skip \Call{w3}{} if \Call{w3}{} has been called
for some earlier event in thread~$i$ and no new critical sections
from some other thread are added to the history.

\subsection{Read-Read Pair Removal}

\begin{algorithm}
\caption{\PWREE\ Read-Read Optimizations }\label{alg:w3poee-rr}
 {\small

\begin{algorithmic}[1]
  \Procedure{read}{$i,x$}
      \State $j = \lastWriteVCt{x}$
  \If {$\accVC{\threadVC{i}}{j} > \accVC{\lastWriteVC{x}}{j} \wedge \LSt{i} \cap \lastWriteVCL{x} = \emptyset$}
    \State $reportPotentialRace(\thread{i}{\accVC{\threadVC{i}}{i}}, \thread{j}{\accVC{\lastWriteVC{x}}{j}})$
  \EndIf
\State $\threadVC{i} = \threadVC{i} \sqcup \lastWriteVC{x}$
\State $\threadVC{i} = \Call{w3a}{\threadVC{i},\LSt{i}}$
\State $\vcEvt = \{ (\thread{i}{\accVC{\threadVC{i}}{i}}, \threadVC{i}, \LSt{i}) \} \cup \vcEvt$
\State $\edges{x} =
       \{ \thread{j}{k} \gtEdge \thread{i}{\accVC{\threadVC{i}}{i}}
       \mid \thread{j}{k} \in \rwVC{x} \wedge
       k < \accVC{\threadVC{i}}{j} \} \cup \edges{x}$
\State $\concEvt{x} = \{ (\thread{j}{k}, \thread{i}{\accVC{\threadVC{i}}{i}})
    \mid \thread{j}{k} \in \rwVC{x}
          \wedge k > \accVC{\threadVC{i}}{j} \wedge \mbox{$\thread{j}{k}$ is a write} \} \cup \concEvt{x}$
\State $\rwVC{x} = \{ \thread{i}{\accVC{\threadVC{i}}{i}} \}
         \cup \{ \thread{j}{k} \mid \thread{j}{k} \in \rwVC{x} \wedge
         (k > \accVC{\threadVC{i}}{j} \vee \mbox{$\thread{j}{k}$ is a write}) \}$
\State $\incC{\threadVC{i}}{i}$
\EndProcedure
\end{algorithmic}
}
\end{algorithm}

The set $\concEvt{x}$ also maintains concurrent read-read pairs.
This is necessary as we otherwise might miss to detect some read-write race pairs.
We give an example shortly.
In practice there are many more reads compared to writes.
Hence, we might have to manage a high number of concurrent read-read pairs.

We can remove all read-read pairs from $\concEvt{x}$ if we relax
the assumptions on $\rwVC{x}$.
Usually, all events in $\rwVC{x}$ must be concurrent to each other.
We relax this condition as follows:
\begin{itemize}
 \item All writes considered on their own and all reads considered on their own are concurrent to each.
 \item A write may happen before a read.
\end{itemize}
Based on the relaxed condition, set $\concEvt{x}$ no longer needs to keep track of read-read pairs.

Algorithm~\ref{alg:w3poee-rr} shows the necessary changes that only
affect the processing of reads.
The additional side condition "$\thread{j}{k}$ is a write"
ensures that no read-read pairs will be added to $\concEvt{x}$.
For $\rwVC{x}$ the additional side condition guarantees that
a write can only be removed by a subsequent write (in happens-before \PWR\ relation).

\fig{f:example}{Read-Read Optimization}{
  \bda{c}
    \ba{lllllllll}
  & \thread{1}{} & \thread{2}{} & \thread{3}{} & \rwVC{x}' & \rwVC{x} & \concEvt{x}' & \concEvt{x} & \edges{x}
  \\ \hline
  1. & \writeE{x} &&& \{ w_1 \} & \{ w_1 \} &&&
  \\ 2. & \readE{x} &&& \{ r_2 \} & \{w_1, r_2\} &&& w_1 \gtEdge r_2
  \\ 3. & & \writeE{x} && \{ r_2, w_3 \} & \{w_1, r_2, w_3\} & (r_2,w_3) & (w_1,w_3)   &
  \\    & &            &&                &                  &           & (r_2, w_3)  &
  \\ 4. & & \readE{x}  && \{ r_2, r_4 \} & \{w_1, r_2, w_3, r_4\} & (r_2,r_4) & (w_1,r_4) & w_2 \gtEdge r_4
  \\ 5. &&& \readE{x} & \{ r_2,r_4,r_5\} & \{w_1, r_2, w_3, r_4, r_5\} & (r_2,r_5) & (w_1, r_5) &
   \\   &&&           &                  &                            & (r_4,r_5) & (w_3, r_5)
   \ea
   \eda
   }

\begin{example}
    Consider the trace in Figure~\ref{f:example}.
  We write $\rwVC{x}'$ and $\concEvt{x}'$ to refer to the sets
  as calculated by Algorithm~\ref{alg:w3poee-firstpass}
  whereas $\rwVC{x}$ and $\concEvt{x}$ refer to the sets
  as calculated by Algorithm~\ref{alg:w3poee-rr}.

  The race pair $(w_1,r_4)$ is detected in the first pass of Algorithm~\ref{alg:w3poee-rr}.
  Based on Algorithm~\ref{alg:w3poee-firstpass} we require some second pass
  to detect $(w_1,r_4)$ based on $w_1 \gtEdge r_2$ and $(r_2,r_4)$.
\end{example}

We conclude. All read-read pairs can be eliminated from $\concEvt{x}$
by making the adjustments described by Algorithm~\ref{alg:w3poee-rr}.
By relaxing the conditions on $\rwVC{x}$ any write-read pair that
is detectable by the second pass via a read-read pair and some write-read edges is
immediately detectable via the set $\rwVC{x}$.
Recall that a write in $\rwVC{x}$ will only be removed from $\rwVC{x}$ if there is
a subsequent write in happens-before \PWR\ relation.
Hence, Algorithms~\ref{alg:w3poee-firstpass} and \ref{alg:w3poee-rr}
and their respective second passes yield the same number of potential race pairs.

The time and space complexities are also the same.
The set $\rwVC{x}$ under the relaxed conditions is still bounded by $O(k)$.
We demand that that all writes considered on their own and
all reads considered on their own
are concurrent to each. Hence, there can be a maximum of $O(k)$ writes and $O(k)$ reads.

The above example suggests that we may also remove write-read edges.
The edge $w_1 \gtEdge r_2$ plays no role for the second pass
based on Algorithm~\ref{alg:w3poee-rr}.
This assumption does not hold in general. The construction
of $\edges{x}$ for Algorithms~\ref{alg:w3poee-firstpass} and ~\ref{alg:w3poee-rr}
must remain the same.

\begin{example}
  Consider the following trace.
  \bda{llll}
  & \thread{1}{} & \thread{2}{} & \thread{3}{}
  \\ \hline
  1. & \writeE{x} &&
  \\ 2. & \readE{x} &&
  \\ 3. & \writeE{y} &&
  \\ 4. && \readE{y} &
  \\ 5. && \readE{x} &
  \\ 6. && \writeE{x} &
  \\ 7. &&& \writeE{x}
  \eda
  Due to the write-read dependency involving variable $y$,
  Algorithm~\ref{alg:w3poee-rr} only reports a single write-write pair, namely $(w_6, w_7)$.
  The additional pair $(w_1, w_7)$ is detected during the second pass
  where write-read and read-write edges such as $w_1 \gtEdge r_2$ and $r_5 \gtEdge w_6$
  are necessary.
\end{example}

\subsection{Aggressive Filtering}

We aggressively apply the filtering check (Lemma~\ref{le:lockset-w3-filtering})
during the second pass.
A pair $(e,f) \in \concEvt{x}$ (step (2) in Definition~\ref{def:all-concs-post})
that fails the Lockset + \PWR\ Filtering check
will not be added to $\accConcs{x}$ (step (4)).
But we have to consider the candidates $(g_i,f)$ and add
them to $\concEvt{x}$ (step (5)) as we otherwise might miss some potential race candidates.

\begin{example}
  Consider the following trace.
  \bda{lll}
  & \thread{1}{} & \thread{2}{}
  \\ \hline
  1. & \writeE{x} &
  \\ 2. & \lockE{y} &
  \\ 3. & \writeE{x} &
  \\ 4. & \unlockE{y} &
  \\ 5. & \writeE{x} &
  \\ 6. && \lockE{y}
  \\ 7. && \writeE{x}
  \\ 8. && \unlockE{y}
  \eda
  In the first pass we obtain $\concEvt{x} = \{ (w_5,w_7) \}$
  and $\edges{x} = \{ w_1 \gtEdge\ w_3 \gtEdge\ w_5 \}$.
  The second pass proceeds as follows.
  Via $(w_5,w_7)$ we obtain the next candidate $(w_3, w_7)$.
  This candidate is not added to $\accConcs{x}$ because locksets of $w_3$
  and $w_7$ are not disjoint. Hence, the filtering check fails.

  We remove $(w_5,w_7)$ from $\concEvt{x}$ but add $(w_3, w_7)$
  to $\concEvt{x}$. Adding $(w_3, w_7)$ is crucial.
  Via $(w_3, w_7)$ we obtain candidate $(w_1, w_7)$.
  This candidate is added to $\accConcs{x}$ (and represents an actual write-write race pair).
\end{example}

There are cases where we can completely ignore candidates.
If the filtering check fails because $e$ and $f$ are in happens-before \PWR\ relation,
then we can completely ignore $(e,f)$ and add $(e,f)$ \emph{not} to $\concEvt{x}$.
This is safe because all further candidates reachable via edge constraints
will also be in \PWR\ relation. Hence, such candidates would fail the filtering check as well.

\subsection{Removal of Critical Sections}

  \begin{algorithm}
  \caption{Thread-local history and removal}\label{alg:w3po-cs-removal}
  \begin{algorithmic}[1]
  \Function{w3}{$i, V, \LStSym$}
    \For {$y \in \LStSym$}
    \For {$(\thread{j}{k}, V') \in \LH{i,y}$}
    \If{\accVC{V'}{j} < \accVC{V}{j}}
      \State $\LH{i,y} = \LH{i,y} - \{(\thread{j}{k}, V') \}$
      \Else
       \If{$k < \accVC{V}{j}$}
       \State $V = V \sqcup V'$
            \EndIf
     \EndIf
     \EndFor
     \EndFor

    \Return V
    \EndFunction
  \end{algorithmic}

\begin{algorithmic}[1]
  \Procedure{release}{$i,y$}
  \State $\threadVC{i} = \Call{w3}{i, \threadVC{i},\LSt{i}}$
  \State $\LSt{i} = \LSt{i} - \{ y \}$
    \For {$i' \not= i$}
  \State $\LH{i',y} = \LH{i',y} \cup \{ (\Acq(x), \threadVC{i}) \}$
  \EndFor
\State $\incC{\threadVC{i}}{i}$
\EndProcedure
\end{algorithmic}

  \end{algorithm}

The history of critical sections for lock $y$ is maintained by $\LH{y}$.
We currently only add critical sections without ever removing them.
From the view of thread~$i$ and its to be processed events,
we can safely remove a critical section
if (a) thread~$i$ has already synchronized with this critical section
(see function \Call{w3}{} in Algorithm~\ref{alg:w3poee-firstpass}), and
(b) the release event happens-before the yet to be processed events.

Removing of critical sections is specific to a certain thread.
Hence, we use thread-local histories $\LH{i,y}$ instead of a global history $\LH{y}$.
Both removal conditions can be integrated into function \Call{w3}{}.
See the updated function \Call{w3}{} in Algorithm~\ref{alg:w3po-cs-removal}.
Function \Call{w3}{} additionally expects the thread id (and therefore all calls must include
now this additional parameter).

We always remove after synchronization.
Hence, removal checks (a) and (b) boil down to the same check
which is carried out within line numbers 4-6.
If the time stamp of the release is smaller compared to thread's time stamp (for the thread
the release is in), the release happens-before and therefore
the critical section can be removed.

In case of a release event, we add the critical section to all other
thread-local histories. Processing of all other events as well as the second pass
remains unchanged.

In theory, the size of histories can still grow considerably.

\begin{example}
  Consider the following trace.
  \bda{lll}
  & \thread{1} & \thread{2}
  \\ \hline
  1. & \lockE{y} &
  \\ 2. & \writeE{x_1} &
  \\ 3. & \lockE{y} &
  \\ \dots & &
  \\ & \lockE{y} &
  \\ & \writeE{x_n} &
  \\ & \unlockE{y} &
  \\ & & \lockE{y}
  \\ & & \readE{x_i}
  \\ & & \unlockE{y}
  \eda
  In thread 2's thread-local history we would find all $n$ critical sections of thread 1.
  This shows that size of thread-local histories may grow linearly in the size of the trace.

  As we assume the number of distinct variables is a constant,
  some of the variables $x_j$ might be repeats.
  Hence, we could truncate thread 2's thread-local history by only keeping
  the most recent critical section that contains a write access to $x_j$.
  Hence, the number of distinct variable imposes a bound on the size
  of thread-local histories.
\end{example}

Similarly, we can argue that the number of thread imposes a bound
on the size of thread-local histories.
Hence, we claim that the size of thread-local histories
be limited to the size $O(v*k)$ without compromising the correctness
of the resulting \PWR\ relation.
We assume that $k$ is the number of threads and $v$ the number of distinct variables.

Maintaining the size $O(v*k)$ for thread-local histories
would require additional management effort.
Tracking thread id's of critical sections and the variable accesses that
occur within critical sections etc.
In our practical experience, it suffices to simply impose a fixed limit
for thread-local histories. For the examples we have encountered,
it suffices to only keep the five most recent critical sections.
That is, when adding a critical section to a thread-local history
and the limit is exceeded, the to be added critical section simply overwrites
the oldest critical section in the thread-local history.

\section{Precision Benchmarks}
\label{sec:precision}

\begin{table*}[t]
{\small
\begin{tabular}{l|c|c|c|c|c}
\textbf{} & $\#$Race Candidates / False Positives & \NoFP\ & \OnlyFP\ & \NoFN\ & \OnlyFN\ \\ \hline
FastTrack & 23 / 5 & 25 & 0  & 4 & 15 \\
SHB & 14 / 0 & 28 & 0 & 4 & 15 \\
\SHBEELimit{} & 19 / 0 & 28 & 0 & 5 & 15 \\
TSan & 54 / 16 & 17 &  5 & 20 & 0 \\
\TSANWRD  & 46 / 8 & 20 & 4 & 20 & 0 \\
\PWRZero{} & 45 / 7 & 21 & 3 & 20 & 0 \\
\PWREELimit{} & 52 / 7 & 21 & 3 & 22 & 0 \\
WCP & 31 / 7 & 23 & 1 & 10 & 9 \\
\end{tabular}
}
\caption{Precision results (28 test cases with 45 predictable races)}
\label{tab:overall-res}
\end{table*}

The precision benchmark suite consists of
28 tests cases that give rise to 45 predictable races.
For 6 out of the 28 test cases there are no data races.
Many test cases require alternative schedules to be explored to predict the the data race.

Recall that \PWREELimit{} employs a limited number of edge constraints which may result in incompleteness (false negatives)
and also limits the history of critical sections which may lead to more false positives.
The limits we employ \PWREELimit{}
have no impact on the number of false negatives
and false positives compared to \PWREE{}.
We introduce the additional candidates SHB and \TSANWRD{}.
SHB is a variant of $\SHBEELimit$ where the limit of edge constraints is zero.
\TSANWRD{} is an extension of TSan that includes write-read dependencies.

Table \ref{tab:overall-res} shows the precision measurements for each algorithm.
Column $\#$Race Candidates / False Positives reports the overall number of race candidates reported
and the number of false positives among candidates.
TSan reports the highest number of race candidates (54) but includes a large number of
false positives (16). Hence, only 38 (=54-16) are (actual) data races.
\TSANWRD{} catches like TSan 38 (=46-8) data races but reports fewer race candidates (46) out of
which eight are false positives.
The precision of \PWRZero{} is similar to \TSANWRD{}.
38 data races are caught out of 45 candidates that include seven false positives.
\PWREELimit{} catches all 45 (=52-7) races and reports 52 race candidates
out of which seven are false positives.
WCP reports 31 race candidates out of which 24 (=31-7) are data races
due to seven false positives.
SHB and \SHBEE{} report the fewest number of race candidates but come with the guarantee
that no false positives are reported.
FastTrack catches 18 (=23-5) data races. Recall that FastTrack ignores write-read dependencies.

Based on the overall precision measured in terms of number of race candidates and false positives,
we draw the following conclusions.
When it comes to zero false positives, SHB and \SHBEE{} perform best.
TSan yields many false positives.
When aiming for many data races with a manageable number of false positives,
\TSANWRD{}, \PWRZero{} and WCP are good choices.
\PWREELimit{} is the best choice when the aim is to catch all data races
with a manageable number of false positives.
FastTrack yields also a manageable number of false positives but catches
considerably fewer data races.

We examine in more the detail the issue of false positives and false negatives.
For this purpose, we measure the number of tests for which an algorithm yields
no false positives among candidates reported (column \NoFP), only false positives among candidates reported (column \OnlyFP),
no false negatives (column \NoFN), only false negatives (column \OnlyFN).
By no false negatives we mean that all races for that test are reported.
By only false negatives we mean that no races are reported although the test has a race.

FastTrack does not report any false positives for 25 out of the 28 tests cases.
See column \NoFP{}. On the other hand, there are 15 tests cases with races for which
no race is reported (column \OnlyFN{}) and there are only 4 test cases for which
all races are reported (column \NoFN{}).
Any case listed in \OnlyFN{} also contributes to \NoFP{}.
Hence, the number 25 in column \NoFP{} results from
the fact that FastTrack reports considerably fewer race candidates compared to some of the other algorithms.
SHB and \SHBEE{} have the same number of ``false negative'' cases as FastTrack.
Their advantage is that both come with the guarantee of not having any false positives.

WCP is able to detect races resulting from alternative schedules.
This is the reason that WCP performs better than FastTrack, SHB and \SHBEE{} when
comparing the numbers in columns \NoFN{} and \OnlyFN{}.
However, WCP appears to be inferior compared to the family of ``TSan'' and ``PWR'' algorithms.
TSan, \TSANWRD{} and \PWRZero{} are able to report all races for 20 test cases.
For \PWREELimit{} we find 22 test cases. See column \NoFN{}.
Recall that there are 28 test cases overall out of which 22 test cases
have races and six test cases have no races.
Hence, 22 test cases is the maximum number to achieve in column \NoFN{}.

In summary, the performance and precision benchmark suites show that
\PWREELimit{} offers competitive performance while achieving high precision.

\section{Counting Data Races}
\label{sec:counting-data-races}

In our measurements, we report 252(2) to indicate that 252 races are found
overall of which two races are found via edge constraints.
Because we only count races modulo their code locations, it is possible
that a race can be found directly \emph{and} via edge constraints.
Instead of ``direct'' races and ``edge constraint'' races, we introduce
the notation of first pass and second pass race.

We refer to a \emph{second pass} race as a race that is
obtained via the help of edge constraints $\edges{x}$
and the set $\concEvt{x}$ of concurrent reads/writes.
We refer to a \emph{first pass} race as a race that is either
a write-read race or a race pair from $\concEvt{x}$ that is not a second pass race pair.
We assume that race pairs are reported following the order as defined by
the construction in Definition~\ref{def:all-concs-post}.

The side condition ``not a second pass race pair'' for a first pass race pair
seems strange as there should not be any mix up between first and second pass
races. However, this is possible for two reasons.
We only report race pairs modulo their code locations
and we compare variants of ``PWR'' that impose different limits on edge constraints.

Consider the trace

 \bda{ll|l}
 & \thread{1}{} & \thread{2}{} \\ \hline
 1. & \writeE{x}^a & \\
 2. & \writeE{x}^b & \\
 3. & & \writeE{x}^c \\
 4. & \writeE{x}^a
 \eda
 where we attach superscripts to indicate the code locations where the events result from.
 We find that $w_1$ and $w_4$ result from the same code location.

 We find that $\concEvt{x} = \{ (w_2,w_3), (w_3,w_4) \}$
 and $\edges{x} = \{ w_1 \gtEdge w_2, w_2 \gtEdge w_4 \}$.
 We first report the first pass race $(w_2,w_3)$.
 Via edge constraints we report the second pass race $(w_1, w_3)$.
 The race pair $(w_3,w_4)$ is not reported because we have already reported
 $(w_1, w_3)$. Recall that these race pairs refer to the same code locations.
 In terms of code locations, we find the first pass race $(b,c)$ and the second pass race $(a,c)$.

 If we impose a limit on edge constraints, say zero, we only consider $\concEvt{x}$.
 Then, first pass races reported are $(w_2,w_3)$ and $(w_3,w_4)$.
 In terms of code locations, we report $(b,c)$ and $(a,c)$.
 But this means that $(a,c)$ can either be reported as a first pass or second pass race
 depending on the limit imposed on edge constraints.

 Such cases arise in our measurements in Table~\ref{tab:realworldbench-singlephase}.
For H2, \PWRZero{} reports 252 first pass races whereas \PWREELimit{} reports 252(2) races.
That is, two of the second pass races are already reported by
\PWRZero{} as first pass races due to the fact the filter races modulo their code locations.

 \end{document}